\pdfoutput=1
\documentclass[%
 reprint,
superscriptaddress,
 amsmath,amssymb,
 aps,
prl,
floatfix,
]{revtex4-2}
\usepackage{graphicx}% Include figure files
\usepackage{dcolumn}% Align table columns on decimal point
\usepackage{bm}% bold math
\usepackage{amsthm}
\usepackage{commath}
\usepackage{mathtools}
\usepackage{graphicx}
\usepackage[colorlinks=true, hyperindex, breaklinks, linkcolor=blue, urlcolor=blue, citecolor=blue]{hyperref} % Physical Review style
\usepackage{amsmath,amsfonts,amssymb,color,graphicx,xcolor,physics}
\usepackage{tikz}
\usepackage{float}
\makeatletter
\let\newfloat\newfloat@ltx
\makeatother
\usepackage{algorithm} 
\usepackage{algorithmicx}
\usepackage[noend]{algpseudocode} 
\newtheorem{theorem}{Theorem}
\newtheorem{lemma}{Lemma}

\newtheorem{corollary}{Corollary}
\newtheorem*{theorem*}{Theorem}
\newtheorem*{corollary*}{Corollary}

\newcommand\nonpfrate[1]{\gamma_{X, Y}}
\makeatletter
\newcommand*{\rom}[1]{\expandafter\@slowromancap\romannumeral #1@}
\makeatother
\let\oldproof\proof
\renewcommand{\proof}{\oldproof}

\def\algbackskip{\hskip-\ALG@thistlm}

\begin{document}

\preprint{APS/123-QED}

\title{The Near-optimal Performance of Quantum Error Correction Codes}

\author{Guo Zheng} 
\thanks{These authors contributed equally.}
\email{guozheng@uchicago.edu}
% \email{guozheng@uchicago.edu}
\affiliation{Pritzker School of Molecular Engineering, The University of Chicago, Chicago 60637, USA}

\author{Wenhao He}
\thanks{These authors contributed equally.}
\affiliation{School of Physics, Peking University, Beijing 100871, China}
\affiliation{Center for Computational Science and Engineering, Massachusetts Institute of Technology, Cambridge, MA 02139, USA}

\author{Gideon Lee}
\affiliation{Pritzker School of Molecular Engineering, The University of Chicago, Chicago 60637, USA}

\author{Liang Jiang}
\email{liang.jiang@uchicago.edu}
\affiliation{Pritzker School of Molecular Engineering, The University of Chicago, Chicago 60637, USA}

%\affiliation{AWS Center for Quantum Computing, Pasadena, CA 91125, USA}


\begin{abstract}

The Knill-Laflamme (KL) conditions distinguish exact quantum error correction codes, and it has played a critical role in the discovery of state-of-the-art codes. However, the family of exact codes is a very restrictive one and does not necessarily contain the best-performing codes. Therefore, it is desirable to develop a generalized and quantitative performance metric. In this Letter, we derive the near-optimal channel fidelity, a concise and optimization-free metric for arbitrary codes and noise. The metric provides a narrow two-sided bound to the optimal code performance, and it can be evaluated with exactly the same input required by the KL conditions. We demonstrate the numerical advantage of the near-optimal channel fidelity through multiple qubit code and oscillator code examples. Compared to conventional optimization-based approaches, the reduced computational cost enables us to simulate systems with previously inaccessible sizes, such as oscillators encoding hundreds of average excitations. Moreover, we analytically derive the near-optimal performance for the thermodynamic code and the Gottesman-Kitaev-Preskill (GKP) code. In particular, the GKP code's performance under excitation loss improves monotonically with its energy and converges to an asymptotic limit at infinite energy, which is distinct from other oscillator codes.

\end{abstract}

\maketitle

%\tableofcontents

% \section{Introduction}

% Oftentimes, it is reasonable to assume some knowledge about the noise processes occurring in the devices. Therefore, it is advantageous to design QEC codes tailored to such noise processes. 

\textit{Introduction.}---Quantum error correction (QEC) has central importance in scaling up quantum devices. The seminal work~\cite{PhysRevA.55.900} by Knill and Laflamme (KL) outlines the necessary and sufficient conditions for exact QEC. The KL conditions are celebrated for their conciseness and computational efficiency. Practically, they have guided the discoveries and analysis of many state-of-the-art qubit codes~\cite{bravyi1998quantum, PhysRevA.52.R2493, PhysRevA.56.2567} and oscillator codes~\cite{PhysRevA.64.012310, PhysRevLett.111.120501, NatureLescanne2020, PhysRevX.6.031006, XuNpjqi2023}. 

The KL conditions deal with \emph{exact} error correction, in that they tell us whether or not any given error is exactly correctable by a given code. However, there are two issues with this: first, the set of correctable errors may not correspond exactly to the errors that occur in real devices -- one often needs to approximate practical noise sources by considering a truncated set of Kraus operators and/or through techniques such as Pauli Twirling~\cite{PhysRevX.11.041039, PhysRevA.94.052325}. Furthermore, codes that \emph{exact} correct leading order errors do not necessarily outperform codes that ``approximately" correct errors at all orders~\cite{PhysRevA.56.2567, PhysRevA.97.032346}. Therefore, it is critical to develop a concise and efficiently computable metric (similar to the KL conditions) that quantifies the capabilities of general codes. Such an extension uncovers the fundamental limit set by the encoding and noise, which is a critical benchmark for code designs.

One such widely accepted benchmark is the channel (or process) fidelity~\cite{PhysRevLett.94.080501, PhysRevA.75.012338, Denys2023tqutrittwomode, PhysRevA.54.2614}\footnote{The channel fidelity is equivalent to the entanglement fidelity for maximally entangled states}. Under this metric, the optimal recovery can be found through convex optimizations, which motivated works in optimization-based QEC~\cite{PhysRevLett.94.080501, PhysRevA.65.030302, Kosut_2009, 8482307, PhysRevA.97.032346, PhysRevA.75.012338, PhysRevA.106.022431, fletcher2007channeladapted, reimpell2006comment, Yamamoto_2005, PhysRevLett.104.120501}. The optimal recovery fidelity can also serve as a guide for encoding designs~\cite{XuZheng2023, Denys2023tqutrittwomode, PhysRevA.97.032346, PhysRevA.106.022431}. However, these methods have two drawbacks. Firstly, although convex optimization algorithms are decently optimized, they remain computationally expensive compared to optimization-free methods. Consequently, past numerical works only work with small Hilbert space sizes, such as systems with fewer than five qubits~\cite{PhysRevLett.100.020502, PhysRevLett.94.080501, PhysRevA.101.042307} or oscillators containing at most ten average excitations~\cite{8482307, PhysRevA.97.032346,Leviant2022quantumcapacity}. Secondly, these optimization techniques are inherently numerical. While powerful, they do not yield analytical forms for many of the quantities we are interested in, eg. the parametric dependence of the code performance. 

Near-optimal recoveries~\cite{10.1063/1.1459754, PhysRevA.81.062342, Petz:1988usv, wilde2017quantum, 10.1063/1.3463451, PhysRevLett.107.080501, PhysRevLett.104.120501} have the potential to circumvent these limitations. These channels have constructive forms and solve a relaxed optimization problem: their performances provide two-sided bounds on the optimal fidelity. These channels have led to attempts to generalize the KL conditions~\cite{PhysRevLett.104.120501, B_ny_2011, PhysRevA.81.062342, PhysRevA.86.012335, Junge_2018}, further leading to the development of codes like the thermodynamic code~\cite{PhysRevLett.123.110502, PhysRevX.10.041018, Zhou2021newperspectives}. However, the generalized conditions still involve optimizations and/or the Bures metric, which are generally challenging to analyze. The common solutions were to derive bounds on the scaling of the near-optimal fidelity assuming large system sizes, which no longer provides a two-sided bound on the optimal performance. Moreover, the parametric dependence on other system parameters remains unknown.

In this Letter, we derive a concise and optimization-free performance metric, the near-optimal channel fidelity. The metric is achievable by the transpose channel~\cite{10.1063/1.1459754}, also known as the Petz channel~\cite{wilde2017quantum, Petz:1988usv, PhysRevLett.128.220502}, and provides a narrow two-sided bound on the optimal fidelity. Crucially, the only input to our expression is the QEC matrix, which is exactly what is required to verify the KL conditions. Therefore, our result is a quantitative generalization of the KL conditions for arbitrary codes and noise processes. We also develop a perturbative expression, providing intuition on how are codes' performances connected to the structures of their error subspaces. More importantly, the perturbative form well approximates the near-optimal fidelity and can be computed analytically. Taking multiple qubit codes as examples, we numerically validate the proximity of the closed-form near-optimal expression and the optimal fidelity obtained from convex optimizations. Furthermore, we analytically compute the near-optimal fidelity for the thermodynamic code in the thermodynamic limit, for which only a scaling with system size was known in past works. After examining a few representative oscillator codes, we provide rigorous insights on why certain codes' performances under excitation loss improve monotonically with increased energy, such as the Gottesman-Kitaev-Preskill (GKP) code~\cite{PhysRevA.64.012310}. While past numerical simulations of the GKP code were limited to a few average excitations, we extend our result to hundreds of excitations. We also obtain GKP's performance analytically, parameterized by system parameters and loss rates. With the analytical expression, we find the GKP's performance admits an asymptotic limit at infinite energy.

% \section{Background}

\textit{Background.}---The Knill-Laflamme (KL) conditions~\cite{PhysRevA.55.900} are the necessary and sufficient conditions for exact QEC codes. For completeness, we briefly review the conditions here. In QEC settings, the logical information is encoded through a code with $d_L$ logical codewords $\set{\ket{\mu_L}}$, which subsequently passes through a noise channel $\mathcal{N}$ with Kraus form $\set{\hat{N}_i}$. The QEC matrix is defined as 
\begin{eqnarray}
    M_{[\mu  l], [\nu k]} = \bra{\mu_L} \hat{N}_l^\dagger \hat{N}_k\ket{\nu_L} \label{eq:QEC_mat}
\end{eqnarray}
in index notation. The KL conditions states that a code is a exact error-correcting code if and only if the QEC matrix can be written as $M = I_L \otimes A$, with $I_L$ denoting a \emph{logical} $d_L$-dimensional identity matrix. 

The KL conditions assess codes with the assumption that any physical recovery is allowed. Such an idea can be extended to general codes: the performance of an encoding, $\mathcal{E}$, against certain noise, $\mathcal{N}$, is determined by the performance of the optimal recovery, $\mathcal{R}^{\text{opt}}$. To define optimality, in this work, we adopt the metric of channel fidelity. For any quantum channel $\mathcal{Q}$, the channel fidelity is defined as~\cite{PhysRevA.75.012338, PhysRevA.54.2614}
\begin{eqnarray}
    F\left(\mathcal{Q}\right) := \bra{\Phi} \mathcal{Q} \otimes\mathcal{I}_R\left(\ket{\Phi}\bra{\Phi}\right)\ket{\Phi},
\end{eqnarray}
where $\ket{\Phi}$ is the purified maximally mixed state, and $\mathcal{I}_R$ is the identity channel acting on the \emph{reference} ancillary system. The optimal recovery, $\mathcal{R}^{\text{opt}}$, is defined as a recovery that achieves the optimal fidelity
\begin{align}
        F^{\text{opt}} := \max_{\mathcal{R}} F\left(\mathcal{R}\circ\mathcal{N}\circ\mathcal{E}\right) = F\left(\mathcal{R}^{\text{opt}}\circ\mathcal{N}\circ\mathcal{E}\right),
\end{align}
where $\circ$ indicates channel compositions. For discussions below, we refer to the channel fidelity as fidelity for simplicity. Our choice of metric is well-motivated by two important properties of the channel fidelity. Firstly, the channel fidelity is directly connected to other widely adopted metrics such as the average input-output fidelity~\cite{PhysRevA.60.1888, Nielsen_2002}. Secondly, the metric is linear in the Choi matrix of the recovery, which causes the optimization to fall in the category of semidefinite programming (SDP)~\cite{SM}, a subfield of convex optimization.

\textit{Main result.}---Here we propose the near-optimal fidelity as a quantitative metric that can be evaluated without optimization.
\begin{theorem}[The near-optimal fidelity]\label{theo:Fe}
    For a $d_L$-dimensional encoding, $\mathcal{E}$, and a noise channel, $\mathcal{N}$, the near-optimal fidelity is
    \begin{align}
        \tilde{F}^{\text{opt}}=\frac{1}{d_L^2}\left\|\operatorname{Tr}_L \sqrt{M}\right\|_F^2,\label{eq:near_opt_F}
    \end{align}
    where $M$ is the QEC matrix, $\left(\operatorname{Tr}_L B\right)_{l,k} = \sum_\mu B_{[\mu l], [\mu k]}$ denotes the partial trace over the code space indices, and $||\cdot||_F$ is the Frobenius norm. The near-optimal fidelity gives a two-sided bound on the optimal fidelity as
    \begin{eqnarray}
    \frac{1}{2}\left(1 - \tilde{F}^{\text{opt}} \right)\leq 1 - F^{\text{opt}} \leq 1 - \tilde{F}^{\text{opt}}.\label{eq:two_sided_bound}
    \end{eqnarray}
\end{theorem}
Worth noticing, the gap between the two-sided bound is proportional to the optimal infidelity. Therefore, when the code performs well against the noise channel, the near-optimal fidelity is a close approximation to the optimal fidelity. Equally importantly, it is remarkable that Eq.~\eqref{eq:near_opt_F} is only dependent on the QEC matrix. Therefore, our metric requires exactly the same resources as the KL conditions, but it provides a quantitative performance metric beyond a binary Yes-or-No output. Moreover, our result can further extended to general channel reversals~\cite{7404264, PhysRevA.101.032109}, subsystem codes~\cite{PhysRevLett.94.180501}, and mixed-state codes~\cite{PhysRevLett.82.2598}. Our result implies that the QEC matrix contains richer information about the code and noise structure beyond the KL conditions. For example, one can greatly reduce the complexity of optimization-based methods via adopting the error subspaces as an efficient basis to describe the action of the noise channel~\cite{SM, mohan2023generalized, johnston2023tight}.

The near-optimal fidelity is achievable by the transpose channel~\cite{10.1063/1.1459754, PhysRevA.81.062342} by construction, which is exactly where our expression inherits the two-sided bound. While there are other channels possessing similar near-optimal properties~\cite{10.1063/1.3463451, PhysRevLett.104.120501}, the transpose channel has the most concise fidelity expression. The derivation of Eq.~\eqref{eq:near_opt_F} is based on the observation that the QEC matrix is the Gram matrix of the error subspaces. When the QEC matrix is invertible, the error subspaces can be orthonormalized by the Gram matrix. In such an orthonormal basis, the transpose channel is equivalent to a measure-and-recover operation, which leads to the expression of Eq.~\eqref{eq:near_opt_F}. The derivation can be generalized to scenarios of degenerate QEC matrix~\cite{SM}.

Computationally, our approach has a drastically reduced cost compared to optimization-based methods. In many cases, the QEC matrix can be analytically computed. Otherwise, it is possible to efficiently compute the matrix depending on the code and noise of interest. In such cases, suppose we are considering $N_K$ number of noise Kraus operators, the cost of evaluating $\tilde{F}^{\text{opt}}$ is $\mathcal{O}\left((d_L N_K)^3\right)$. As a reference, the SDP for optimal recovery costs $\tilde{\mathcal{O}}\left((d_L N)^{5.246}\right)$~\cite{9317892}\footnote{$\tilde{\mathcal{O}}(g(n))$ is shortened for $\mathcal{O}(g(n)\log^k g(n))$}, where $N$ is the physical Hilbert space dimension. For example, if we consider qubit codes, $N$ scales exponentially with the number of qubits, $n$. However, it is sufficient numerically to truncate the number of noise Kraus operators to a polynomial scaling, $N_K\propto n^r$. In many cases, $r$ only depends on the target precision and physical error rates.

While the exact form of the near-optimal fidelity, Eq.~\eqref{eq:near_opt_F}, is an elegant expression, its advantage lies in its numerical complexity. The component of matrix square root makes it cumbersome to obtain an analytical expression for the near-optimal fidelity. Therefore, we develop the following corollary based on a perturbative decomposition of the QEC matrix.
% \paragraph{Perturbative form of the near-optimal fidelity.} Eq.~\eqref{eq:F^TC} is an elegant expression of the near-optimal fidelity. However, its advantage compared to other optimization-based approaches mostly lies in the numerical complexity, and it is not obvious how to analytically obtain the near-optimal fidelity. Following a similar spirit as the KL conditions, we can decompose the channel reversal matrix into 
\begin{corollary}\label{coro:Fe_pert}
     The noise channel's Kraus representation can be chosen such that $\frac{1}{d_L}\operatorname{Tr}_L M =D$, with $M$ being the QEC matrix and $D$ being a diagonal matrix. With the residual matrix $\Delta M := M - I_L \otimes D$, the near-optimal infidelity has a perturbative expansion through
    \begin{eqnarray}
    &&1 - \tilde{F}^{\text{opt}} \nonumber\\
    &=& \frac{1}{d_L}\norm{f(D)\odot \Delta M }_F^2 + \mathcal{O}\left(\frac{1}{d_L}\norm{f(D)\odot \Delta M }_F^3\right)\label{eq:F^TC_pert}
    \end{eqnarray}
    where $f(D) _{[\mu l],[\nu k]} = \frac{1}{\sqrt{D_{ll}} + \sqrt{D_{kk}}}$ and the Hadamard product $\left(A\odot B\right)_{ij} = A_{ij}B_{ij}$.
\end{corollary}
This corollary is proved with the Daleckii-Krein theorem of matrix square root expansions~\cite{SM, DaletskiiKrein1965, carlsson2018perturbation}. Eq.~\eqref{eq:F^TC_pert} conveniently expresses the infidelity as a function of $\Delta M$ and $D$ instead of the square root of $M$, thus making it tangible to obtain analytical expressions. Physically, $\frac{1}{d_L}I_L\otimes\operatorname{Tr}_L M$ and $\Delta M$ corresponds to the correctable and uncorrectable QEC matrix respectively. In Corollary.~\ref{coro:Fe_pert}, we applied a unitary to diagonalize the correctable matrix. However, it is generally nontrivial to analytically express the unitary. As a compromise, $D$ can be instead defined as the diagonal entries of the correctable matrix, $\text{diag}(\frac{1}{d_L}\operatorname{Tr}_L M) = \text{diag}(D)$. Such a truncation overestimates the error, but its effect is negligible as long as the off-diagonal yet correctable elements are sufficiently small~\cite{SM}. 
% For simplicity, we assume we work with $D = \frac{1}{d}\operatorname{Tr}_L M$ in the following without specifications.

It is important to note that the infidelity includes contributions from the uncorrectable matrix modulated by a function of the correctable matrix, $f(D)$. Intuitively, while the overlaps between error subspaces lead to uncorrectable errors, their effects are weighted by the probabilities of their respective quantum trajectories, which are contained in $D$. Past works~\cite{PhysRevA.97.032346, olle2023simultaneous, PhysRevLett.131.050601, PhysRevA.106.022431, korolev2023error, Cao_2022} have proposed code performance estimators based on the QEC matrix to perform efficient code optimizations. Nonetheless, they mostly only took into account the uncorrectable matrix, $\Delta M$, or attempted to consider the effects of $D$ in heuristic ways. As a comparison, Eq.~\eqref{eq:F^TC_pert} suggests a combination of effects from $\Delta M$ and $D$ with guaranteed performance.

A few observations can be drawn from Eq.~\eqref{eq:F^TC_pert}. For example, for the near-optimal and optimal fidelity, the error caused by the uncorrectable matrix is suppressed quadratically. While the quadratic scaling was also observed in Ref.~\cite{PhysRevLett.107.080501}, our result is not limited to one-parameter family of channels and instead presents the full expression. Moreover, Eq.~\eqref{eq:F^TC_pert} has a noteworthy property: if $\Delta M$ is traceless, the code is an exact QEC code when the perturbative form vanishes~\cite{SM}. This is useful for applications that require vanishing error probability, such as computing the achievable rates~\cite{PhysRevA.64.062301}.

\begin{figure}[t]
    \centering
    \includegraphics[width = 0.45 \textwidth]{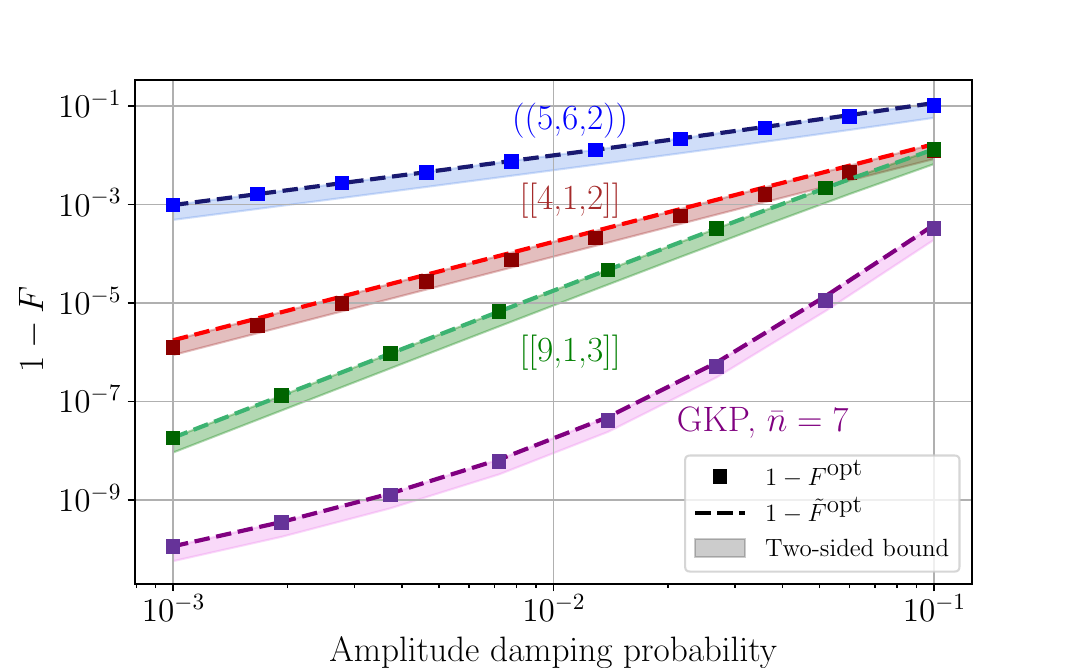}
    % \captionsetup{justification=justified}
     \caption{Optimal infidelity for qubit codes $[[4, 1, 2]]$, $((5, 6, 2))$, and $[[9, 1, 3]]$ and the GKP code with $\bar{n}=7$. The noise channel is amplitude damping noise, i.e. loss for oscillators. The shaded regions represent the optimal infidelity intervals bounded by the two-sided bound given by the near-optimal infidelity in Eq.~\eqref{eq:two_sided_bound}.}
    \label{fig:verification_qubit}
\end{figure}

\textit{Examples.}---In Fig.~\ref{fig:verification_qubit}, we numerically validate the two-sided bound presented in Eq.~\eqref{eq:two_sided_bound} for qubit codes under amplitude damping noise~\cite{SM}, which is a practical but non-Pauli noise channel. We adopt the convention that $[[n, k, d]]$ represents encoding $k$ logical qubits in $n$ physical qubits with distance $d$, while for $((n, k, d))$, $k$ represents the logical dimension instead. Even for well-studied codes like stabilizer codes, their optimal decoders under non-Pauli noise are in general unknown. The shown qubit codes include the classic $[[9, 1, 3]]$ stabilizer code~\cite{1996Steane} and approximate $[[4, 1, 2]]$~\cite{PhysRevA.56.2567} code. Another example is the $((5, 6, 2))$ code~\cite{PhysRevLett.79.953}, which is a qudit non-stabilizer code. 

\begin{figure}[t]
    \centering
    \includegraphics[width = 0.45 \textwidth]{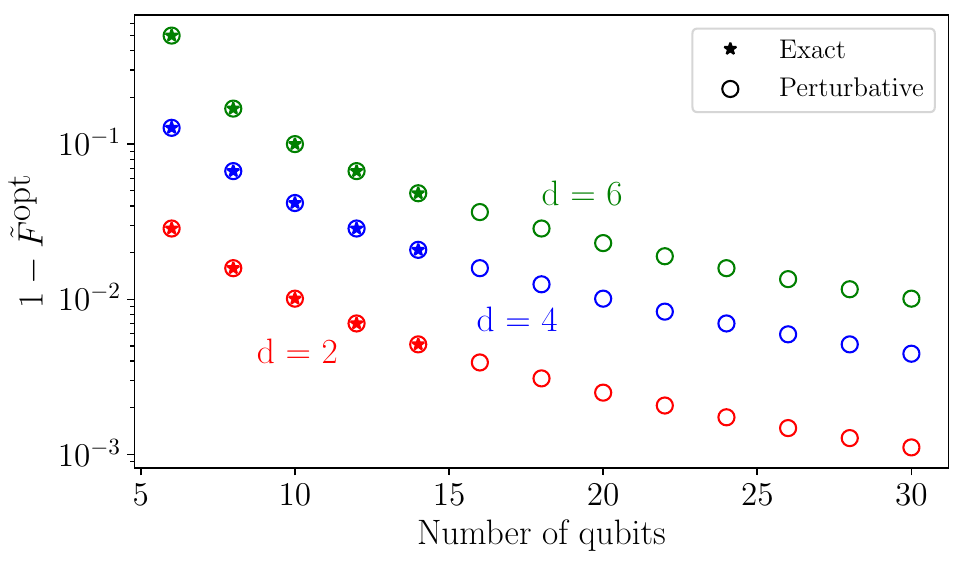}
    %  \caption{
    %  Comparison of the original cat gate scheme with the modified gate scheme with SC encoding of optimized $T_{cool}, T_{gate}, r$ while fixing $\bar{n} = 4$. The inset diagram is a zoomed-in plot of the SC gates' errors around their optimal points. (a) Z-error of a CX gate. (b) Z-error of a Z rotation of angle $\pi$.}
      \caption{The near-optimal infidelity of the thermodynamic code with distance $d$ (see Ref.~\cite{SM} for the definition of the codewords) and $l=1$, i.e. single erasure error. The stars represent the exact approach where the QEC matrix is numerically computed, and the circles represent the perturbative expression given in Eq.~\eqref{eq:thermo_pert}.}
    \label{fig:thermodynamic_code}
\end{figure}

The contributions of our expression lie not only in numerics but also in analytical aspects. We take the thermodynamic code~\cite{PhysRevLett.123.110502, SM} as an example. The code has sparked interest because of its close connections with the Eigenstate Thermalization Hypothesis~\cite{PhysRevE.50.888}, as well as being an instance of a covariant code~\cite{PhysRevX.10.041018}. The precise definition of the codewords are given in Ref~\cite{SM}, and they are characterized by the distance $d$. When considering a constant number of erasure errors and thermodynamic limit, the optimal infidelity was proven to scale as $1 - F^{\text{opt}} = \Theta(\frac{1}{N^2})$~\cite{PhysRevX.10.041018}, where $N$ is the number of qubits. Beyond a scaling argument, our expression makes it possible to derive the near-optimal infidelity analytically~\cite{SM},
\begin{eqnarray}
    1 - \tilde{F}^{\text{opt}} = \frac{l}{16} \frac{d^2}{N^2} + \mathcal{O}\left(\frac{1}{N^3}\right)\label{eq:thermo_pert}.
\end{eqnarray}
% \textcolor{red}{The appearance of the distance $d$ in Eq.~(\ref{eq:thermo_pert}) is important, as it implies that the thermodynamic code cannot tolerate erasure with constant probability: for $\mathcal{O}(n)$ erasure errors and $\mathcal{O}(n)$ distance, the infidelity is a constant independent of system size.} 
for $l$ erasure errors. The derivation of Eq.~(\ref{eq:thermo_pert}) with our formalism is straightforward and can be easily extended to consider more erasure errors~\cite{SM}. The comparison of the perturbative and the exact forms of the near-optimal infidelity is shown in Fig.~\ref{fig:thermodynamic_code}, where it is clear they closely follow each other. For the exact form, the QEC matrix is numerically computed, and the simulation stops at 14 qubits because of the increased computational cost. To compare, if we attempt to optimize for the optimal fidelity, it is only possible for less than 5 qubits under the same time constraint.

\begin{figure*}[t]
    \centering
    \includegraphics[width = 0.9 \textwidth]{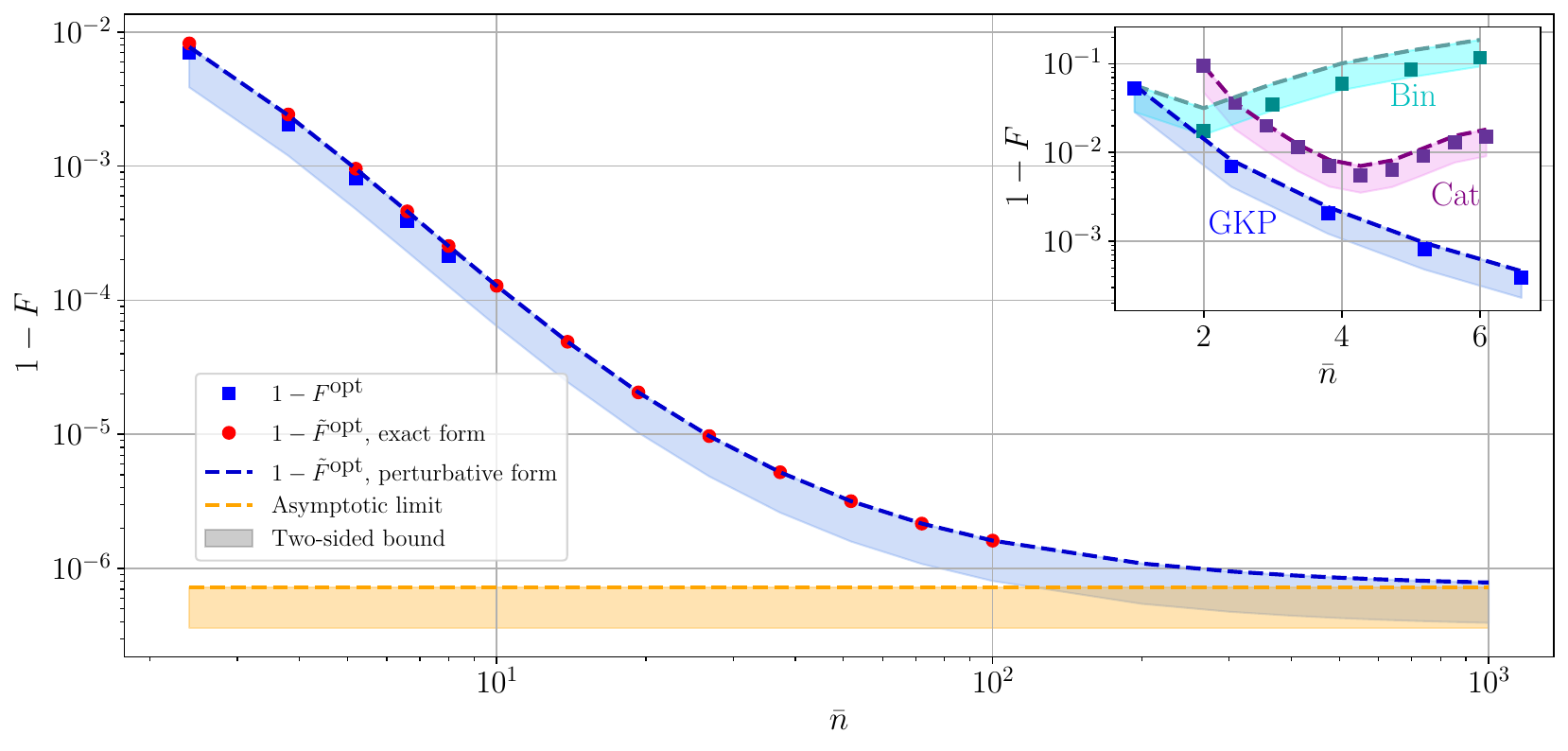}
    %  \caption{
    %  Comparison of the original cat gate scheme with the modified gate scheme with SC encoding of optimized $T_{cool}, T_{gate}, r$ while fixing $\bar{n} = 4$. The inset diagram is a zoomed-in plot of the SC gates' errors around their optimal points. (a) Z-error of a CX gate. (b) Z-error of a Z rotation of angle $\pi$.}
     \caption{The channel infidelity of square lattice GKP qubit code with $\gamma = 10\%$ loss. All shaded regions represent the two-sided bound on the optimal fidelity. The red circles represent the exact form of the near-optimal infidelity, where the QEC matrix is computed analytically. The blue dashed line is the perturbative form, evaluated through a closed-form analytical expression~\cite{OtherPaper}. The inset figure shows the performance of GKP codes (square lattice), cat codes (S=4), and binomial codes (S=1)~\cite{SM}. The squares (dashed lines) represent the optimal (near-optimal) fidelity.}
    \label{fig:gkp_bosonic_codes}
\end{figure*}

% \begin{figure}[t]
%     \centering
%     \includegraphics[width = 0.45 \textwidth]{combined_new_3.pdf}
%     %  \caption{
%     %  Comparison of the original cat gate scheme with the modified gate scheme with SC encoding of optimized $T_{cool}, T_{gate}, r$ while fixing $\bar{n} = 4$. The inset diagram is a zoomed-in plot of the SC gates' errors around their optimal points. (a) Z-error of a CX gate. (b) Z-error of a Z rotation of angle $\pi$.}
%      \caption{\textcolor{red}{Added.} The channel infidelity of square lattice GKP qubit code under loss. The loss rate is $\gamma = 0.1$. The blue dashed lines with (without) squares represent the near-optimal infidelity upper (lower) bound given by the optimal infidelity. The red triangles represent the exact form of the near-optimal fidelity, where the QEC matrix is computed analytically. The green line is the perturbative form, \textcolor{red}{which has a closed-form analytical expression}. The inset figure shows the optimal channel infidelity of GKP codes (square lattice), cat codes (S=4), and binomial codes (S=1)~\cite{SM}.}
%     \label{fig:gkp_bosonic_codes}
% \end{figure}

Bosonic codes, or oscillator codes, are codes that encode a $d_L$-dimensional logical space in the infinite-dimensional Hilbert space of oscillator(s). For a bosonic code, the average excitation number, $\bar{n} := \Tr\left(\hat{n}\hat{P}_L\right)$, affects code properties and is also a parameter of experimental interest. In the inset plot of Fig.~\ref{fig:gkp_bosonic_codes}, we show the optimal performance of a few popular bosonic codes under excitation loss~\cite{SM}, including the binomial code~\cite{PhysRevX.6.031006}, the cat code~\cite{NatureLescanne2020}, and the GKP code~\cite{PhysRevA.64.012310}. An emerging feature is that for the binomial and cat codes, their performance do not improve monotonically with energy with fixed $S$. Here, $S$ represents the fock basis spacing, which can be understood as the distance against loss. On the contrary, the GKP code's performance improves monotonically for the range of $\bar{n}$ shown. Similar numerical observations were made in previous works~\cite{PhysRevA.97.032346, PhysRevX.6.031006}, where heuristic explanations were given based on the QEC matrix. Our expression Eq.~\eqref{eq:F^TC_pert} supports their arguments with rigor: like exact qubit codes, the cat code and binomial code exactly correct no more than $S$ excitation loss. However, they are completely unprotected against more than $S$ loss, which is more likely to occur at larger $\bar{n}$. Therefore, the infidelity is not suppressed with energy. 

The critical difference of GKP codes is that while it generally cannot even correct a single loss, the uncorrectable elements for all orders of loss are suppressed by a factor of $e^{-\frac{\gamma}{\gamma + \frac{1}{\bar{n}}}}$~\cite{PhysRevA.97.032346}. One open question was whether the suppression of the infidelity holds asymptotically: optimizing for the optimal fidelity was only possible for $\bar{n}\leq 10$~\cite{PhysRevA.97.032346} because the Hilbert space size quickly becomes unmanageable. The near-optimal fidelity solves the problem. Firstly, with the QEC matrix being analytically computable, the only complexity cost for the exact form in Eq.~\eqref{eq:near_opt_F} is to compute the matrix square root. Thus, it is much cheaper than SDP optimizations, and the numerical results can reach $\bar{n}\sim 10^2$. Secondly, we can express the near-optimal performance analytically with the perturbative form~\cite{OtherPaper}, which reveals the near-optimal fidelity for arbitrarily large $\bar{n}$. In Fig.~\ref{fig:gkp_bosonic_codes}, we demonstrate the results for a GKP square code under loss rate $\gamma=0.1$. At asymptotically large $\bar{n}$, the perturbative expression converges as
\begin{eqnarray}
    \lim_{\bar{n}\to \infty} 1 - \tilde{F}^{\text{opt}} = e^{-\frac{\pi}{2}\frac{1-\gamma}{\gamma}},
\end{eqnarray}
which is approximately $7.2\times 10^{-7}$ for $\gamma=0.1$. As a comparison, the best decoder with known performance at infinite energy is the amplification decoder (AD)~\cite{PhysRevA.97.032346, 8482307}. For the same loss, AD gives a logical error rate of $1 - F^{\text{AD}} \approx e^{-\frac{\pi}{8}\frac{1-\gamma}{\gamma}} = 2.9\times 10^{-2}$. Therefore, while the exceptionally low near-optimal fidelity highlights the promises of the GKP code construction, its gap with the performance of the known decoders emphasizes the potential gain in improving GKP decoders. Similar analysis can be performed for general multi-mode GKP encodings~\cite{OtherPaper}.

\textit{Discussion.}---We have derived the near-optimal channel fidelity as a quantitative metric for arbitrary codes and noise channels. Our metric is closely related to optimal code performances through narrow two-sided bounds. Since the near-optimal fidelity only requires the QEC matrix as input, it generalizes the KL conditions beyond distinguishing exact QEC codes with the same resource cost. To conclude, the proposed metric and its perturbative form reduce computational costs for numerical simulations and enable us to obtain analytical descriptions of code performances.

Our new approach opens doors to numerical benchmarking of large-sized codes and oscillators encoding high-energy states. The benchmarking result can, in turn, be used to optimize or guide the discovery of novel encoding and efficient decoding for realistic noises. Moreover, it is critical to understand the near-optimal performance of codes in asymptotic limits for many concepts and settings in quantum information theory, such as the achievable rate of a code family. 

\begin{acknowledgments}
We thank Victor Albert, Kyungjoo Noh, Sisi Zhou, and Qian Xu for helpful discussions. We acknowledge support from the ARO(W911NF-23-1-0077), ARO MURI (W911NF-21-1-0325), AFOSR MURI (FA9550-19-1-0399, FA9550-21-1-0209), AFRL (FA8649-21-P-0781), DoE Q-NEXT, NSF (OMA-1936118, ERC-1941583, OMA-2137642), NTT Research, and the Packard Foundation (2020-71479). This material is based upon work partially supported by the U.S. Department of Energy, Office of Science, National Quantum Information Science Research Centers. The authors are also grateful for the support of the University of Chicago Research Computing Center for assistance with the numerical simulations carried out in this paper.
\end{acknowledgments}

% \begin{acknowledgments}
% Acknowledge
% \end{acknowledgments}


\clearpage

% \pagebreak
% \break

\widetext
\begin{center}
\textbf{\large Supplemental Material: "The near-optimal performance of quantum error correction codes"}
\end{center}
%%%%%%%%%% Merge with supplemental materials %%%%%%%%%%
%%%%%%%%%% Prefix a "S" to all equations, figures, tables and reset the counter %%%%%%%%%%
\setcounter{equation}{0}
\setcounter{figure}{0}
\setcounter{table}{0}
\setcounter{page}{1}
\makeatletter
\renewcommand{\theequation}{S\arabic{equation}}
\renewcommand{\thefigure}{S\arabic{figure}}
\renewcommand{\bibnumfmt}[1]{[#1]}
\renewcommand{\citenumfont}[1]{#1}

% \appendix 
% \newpage

\section{Convex optimization for recovery}

The problem setup for channel reversal is that we want to revert a general channel $\mathcal{A}: \mathcal{B}\left(\mathcal{H}_{d_L}\right)\to\mathcal{B}\left(\mathcal{H}_n\right)$ with a physical recovery, $\mathcal{R}: \mathcal{B}\left(\mathcal{H}_n\right)\to\mathcal{B}\left(\mathcal{H}_{d_L}\right)$. Then, we would want to minimize the distance between the composed channel, $\mathcal{R}\circ \mathcal{A}$, and the identity channel. One of the distance metrics is the channel fidelity. For a quantum channel $\mathcal{Q}: \mathcal{B}\left(\mathcal{H}_{d_L}\right)\to\mathcal{B}\left(\mathcal{H}_{d_L}\right)$ with an operator sum representation, $\mathcal{Q}\left(\rho\right) = \sum_{i}Q_i \rho Q_i^\dagger$, the channel fidelity is defined as
\begin{eqnarray}
    F\left(\mathcal{Q}\right) &:=& \bra{\Phi} \mathcal{Q} \otimes\mathcal{I}_R\left(\ket{\Phi}\bra{\Phi}\right)\ket{\Phi}\\
    &=& \frac{1}{d_L^2}\sum_i \abs{\text{Tr}\hat{Q}_i}^2 \label{eq:ent_fid_expr}\\
    &=& \frac{1}{d_L^2} \text{Tr} X_{\mathcal{Q}}
\end{eqnarray}
where $\ket{\Phi}$ is the purified maximally mixed state, $\mathcal{I}_R$ is the identity channel acting on the \emph{reference} ancillary system, $d$ is the input and output Hilbert space size, and $X_{\mathcal{Q}}$ is the superoperator form~\cite{Caves1999} (also known as the matrix representation~\cite{nielsen_chuang_2010}) of $\mathcal{Q}$. If $\mathcal{Q} = \mathcal{R}\circ \mathcal{A}$, a nice property of the channel fidelity is that it can be expressed as a linear function of the Choi matrix form of $\mathcal{R}$:
\begin{eqnarray}
    F\left(\mathcal{R}\circ \mathcal{A}\right) &=& \frac{1}{d_L^2} \text{Tr} \left(X_{\mathcal{R}}X_{\mathcal{A}}\right) = \frac{1}{d_L^2} \text{Tr} \left(C_\mathcal{R} \tilde{C}_{\mathcal{A}}\right)
\end{eqnarray}
where the Choi matrix $C_\mathcal{M}$ of any channel $\mathcal{M}$ is defined to have an index form of $\left(C_\mathcal{M}\right)_{[\mu\rho], [\nu\sigma]}= \bra{\rho}\mathcal{M}\left(\ket{\mu_{L}}\bra{\nu_{L}}\right)\ket{\sigma}$. We also define the linear map $\left(\tilde{C}_\mathcal{M}\right)_{[\mu\rho], [\nu\sigma]} = \left(C_\mathcal{M}\right)_{[\sigma\nu], [\rho\mu]}$ to make the linear dependence of $F$ on $C_{\mathcal{R}}$ more apparent. Moreover, the constraints on the recovery channel being physical, i.e. a completely positive trace-preserving map, can also be expressed through linear functions of $C_{\mathcal{R}}$. As a result, we can formulate the channel reversal problem as a convex optimization problem. The optimal recovery $\mathcal{R}^{\text{opt}}$ and its performance can be found through
\begin{eqnarray}
    \max_{C_\mathcal{R}}& \Tr\left\{ C_\mathcal{R} \tilde{C}_{\mathcal{A} }\right\}, \\
    \text{s.t. }& C_\mathcal{R}\succeq 0, \Tr_{\mathcal{H}_n} C_\mathcal{R} = I_{d_L}
\end{eqnarray}
which is a semidefinite program (SDP). Here, the partial trace notation $\Tr_{\mathcal{H}_n}$ denotes tracing out the indices of the input system, and $I_d$ is a $d_L\times d_L$ identity matrix. The first condition ensures the recovery to be a completely positive map, while the second represents the constraint of trace-preserving maps. 

In the context of quantum error correction (QEC), the target channel to be reverted is composed of the noise channel and the encoding channel, $\mathcal{A} = \mathcal{N} \circ \mathcal{E}$. The scenario considered by the original Knill-Laflamme conditions is the special case of pure state encodings, where $\mathcal{E}$ is an isometry. The optimization result specifies the optimal recovery and the optimal channel fidelity, $F^{\text{opt}} = F\left(\mathcal{R}^{\text{opt}}\circ\mathcal{N}\circ\mathcal{E}\right)$. Similarly, we can utilize the permutation invariant property of the trace and similarly optimize the encoding. Since we optimize for the Choi matrix of the recovery and/or encoding, the computational complexity is a high-order polynomial or even exponential function of the system size.

% \textcolor{red}{Should we mention this?} 
Worth noticing, the Fock state basis may not be the most efficient basis for representing the Choi matrix and, more generally, the code after noise. We can learn from the transpose channel form and devise better forms to drastically reduce the complexity required. Some recent works on antidistinguishability~\cite{johnston2023tight} and quantum state discrimination~\cite{mohan2023generalized} adopted similar ideas to improve their SDP forms, and the improvement can be also extended to QEC.

\section{Transpose channel}

The transpose channel (TC)~\cite{10.1063/1.1459754, PhysRevA.81.062342} is a member of the near-optimal channel family. It has close connections with pretty good measurements~\cite{1975RaEl...20.1177B, doi:10.1137/1123048, Hausladen1994AG} and has applications in quantum information theory~\cite{PhysRevA.77.034101}, quantum thermodynamics~\cite{10.1116/5.0060893, PhysRevX.9.031029, Buscemi_2021}, and many others. There is also active effort in devising implementation schemes. There have been efforts concerning the systematic implementation of general CPTP maps in both bosonic~\cite{PhysRevB.95.134501} and qubit systems~\cite{PhysRevA.95.052316}. They can be fault-tolerant through techniques such as path-independent gates~\cite{Rosenblum_2018, PhysRevLett.125.110503}. Focusing on the transpose channel, Ref.~\cite{PhysRevLett.128.220502} proposed an algorithm to implement a Petz channel based on quantum singular value transformation, block encoding, and a set of unitary gates. Moreover, Ref.~\cite{biswas2023noiseadapted} proposed and compared several implementation schemes of the transpose channel on practical devices.

For a given logical code space projector $\hat{P}_L$ and noise channel $\mathcal{N}\sim\set{\hat{N}_i}$, the transpose channel admits a Kraus operator representation of $\mathcal{R}^{\text{TC}}\sim \set{\hat{R}_i^{\text{TC}}}$, where 
\begin{eqnarray}
    \hat{R}^{\text{TC}}_i := \hat{P}_L \hat{N}_i^\dagger \mathcal{N}\left(\hat{P}_L\right)^{-\frac{1}{2}}
\end{eqnarray}
for $i = 0, \dots, L-1$. Here, we define $\hat{P}_L$ to be the projector onto the logical code space. The inverse of $\mathcal{N}\left(\hat{P}\right)$ is defined on its support, which is spanned by the error subspaces. 

The transpose channel has a critical property: its fidelity gives a two-sided bound of the optimal recovery fidelity. While there exist other recovery channels with similar two-sided bounds, such as the quadratic recovery~\cite{10.1063/1.3463451}, the near-optimal expression we obtained from the transpose channel is the most concise and informative by far. For completeness, we provide a proof for the near-optimality.

\begin{lemma} [Near-optimality of transpose channel; \cite{10.1063/1.1459754}, Theorem 2] For any QEC code with encoding channel $\mathcal{E}$ and noise channel $\mathcal{N}$, the optimal channel fidelity $F^{\mathrm{opt}}=F\left(\mathcal{R}^{\mathrm{opt}} \circ \mathcal{N} \circ \mathcal{E}\right)$ is bounded by the channel fidelity of the transpose channel $F\left(\mathcal{R}^{\mathrm{TC}} \circ \mathcal{N} \circ \mathcal{E}\right)$ through
    \begin{align}
            \frac{1}{2}\left(1 - F(\mathcal{R}^{\text{TC}}\circ\mathcal{N} \circ \mathcal{E}) \right)\leq 1 - F^{\text{opt}} \leq 1 -  F(\mathcal{R}^{\text{TC}}\circ\mathcal{N} \circ \mathcal{E}).
    \end{align}
\end{lemma}
\begin{proof}
The second inequality holds with definition of optimal recovery. To prove the first inequality, notice that the recovery should have its domain on $\text{supp}(\mathcal{N}\hat{\rho})$ and its range on $\text{supp}(\hat{\rho})$. The optimal recovery channel $\mathcal{R}^{\text{opt}}\sim \{R_i^{\text{opt}}\}$ can be represented as 
\begin{align}
    \hat{R}_i^{\text{opt}} = \hat{P}_L \hat{B}_i^{\dagger} \mathcal{N}(\hat{P}_L)^{-1 / 2}.
\end{align}
where we can choose $\hat{B}_i^\dagger = \hat{P}_L \hat{R}_i^{\text{opt}} \mathcal{N}(\hat{P}_L)^{1/2}$. The channel fidelity can be written as
    \begin{align}
        F(\mathcal{R}^{\text{opt}}\circ\mathcal{N} \circ \mathcal{E}) &= \frac{1}{d_L^2} \sum_{i, j}\left|\operatorname{Tr}\left(\hat{R}_i^{\mathrm{opt}} \hat{N}_j\right)\right|^2 \\
        &=  \frac{1}{d_L^2} \sum_{i, j}\left|\operatorname{Tr}\left(\hat{P}_L \hat{B}_i^{\dagger} \mathcal{N}\left(\hat{P}_L\right)^{-\frac{1}{2}}\hat{N}_j\right)\right|^2 \\
        &=  \frac{1}{d_L^2} \sum_{i, j}\left|\operatorname{Tr}\left(\hat{P}_L \hat{B}_i^{\dagger} \mathcal{N}\left(\hat{P}_L\right)^{-\frac{1}{2}}\hat{N}_j\hat{P}_L\right)\right|^2 \\
    \end{align}
    By appropriate choice of Kraus operators $\hat{B}_i$ and $\hat{N}_j$, the inner sum can be diagonalized such that the trace vanishes for $i\neq j$, which is equivalent to performing singular value decomposition. Then, we have
    \begin{eqnarray}
         &&F(\mathcal{R}^{\text{opt}}\circ\mathcal{N} \circ \mathcal{E}) \\
         &=& \frac{1}{d_L^2} \sum_{i}\left|\operatorname{Tr}\left(\hat{P}_L \hat{B}_i^{\dagger} \mathcal{N}\left(\hat{P}_L\right)^{-\frac{1}{2}}\hat{N}_i\hat{P}_L\right)\right|^2 \\
        &=&  \frac{1}{d_L^2} \sum_{i}\left|\operatorname{Tr}\left((\hat{P}_L \hat{B}_i^{\dagger} \mathcal{N}\left(\hat{P}_L\right)^{-\frac{1}{4}})(\mathcal{N}\left(\hat{P}_L\right)^{-\frac{1}{4}}\hat{N}_i \hat{P}_L )\right)\right|^2 \\
        &\le& \frac{1}{d_L^2} \sum_{i, j}\left|\operatorname{Tr}\left((\hat{P}_L \hat{B}_i^{\dagger} \mathcal{N}\left(\hat{P}_L\right)^{-\frac{1}{4}})(\mathcal{N}\left(\hat{P}_L\right)^{-\frac{1}{4}}\hat{B}_i \hat{P}_L )\right)\operatorname{Tr}\left((\hat{P}_L \hat{N}_j^{\dagger} \mathcal{N}\left(\hat{P}_L\right)^{-\frac{1}{4}})(\mathcal{N}\left(\hat{P}_L\right)^{-\frac{1}{4}}\hat{N}_j \hat{P}_L )\right)\right|   \\
        &\le&  \frac{1}{d_L^2} \sqrt{\sum_i \left|\operatorname{Tr}\left((\hat{P}_L \hat{B}_i^{\dagger} \mathcal{N}\left(\hat{P}_L\right)^{-\frac{1}{4}})(\mathcal{N}\left(\hat{P}_L\right)^{-\frac{1}{4}}\hat{B}_i \hat{P}_L )\right)\right|^2 \sum_{j}\left|\operatorname{Tr}\left((\hat{P}_L \hat{N}_j^{\dagger} \mathcal{N}\left(\hat{P}_L\right)^{-\frac{1}{4}})(\mathcal{N}\left(\hat{P}_L\right)^{-\frac{1}{4}}\hat{N}_j \hat{P}_L )\right)\right|^2 } \\
        & \le& \frac{1}{d_L^2} \sqrt{\sum_{ik} \left|\operatorname{Tr}\left(\hat{P}_L \hat{B}_i^{\dagger} \mathcal{N}\left(\hat{P}_L\right)^{-\frac{1}{2}}\hat{B}_k \hat{P}_L \right)\right|^2 \sum_{jl}\left|\operatorname{Tr}\left(\hat{P}_L \hat{N}_j^{\dagger} \mathcal{N}\left(\hat{P}_L\right)^{-\frac{1}{2}}\hat{N}_l \hat{P}_L \right)\right|^2 } \label{eq:prof_neear_opti}
    \end{eqnarray}
    where the second term inside square root can be simplified as
    \begin{align}
        \frac{1}{d_L^2}\sum_{jl}\left|\operatorname{Tr}\left(\hat{P}_L \hat{N}_j^{\dagger} \mathcal{N}\left(\hat{P}_L\right)^{-\frac{1}{2}}\hat{N}_l \hat{P}_L \right)\right|^2 = F\left(\mathcal{R}^{\mathrm{TC}} \circ \mathcal{N} \circ \mathcal{E}\right).\label{eq:prof_neear_opti1}
    \end{align}
    The first term can be upper bounded as
    \begin{align}\label{eq:prof_neear_opti2}
        \frac{1}{d_L^2}\sum_{ik} \left|\operatorname{Tr}\left(\hat{P}_L \hat{B}_i^{\dagger} \mathcal{N}\left(\hat{P}_L\right)^{-\frac{1}{2}}\hat{B}_k \hat{P}_L \right)\right|^2 = \frac{1}{d_L^2}\sum_{ik} \left|\operatorname{Tr}\left(\hat{R}^{\text{opt}}_i\hat{B}_k  \right)\right|^2 = F\left(\mathcal{R}^\text{opt}\circ\mathcal{B}\right)\leq 1.
    \end{align}
    The proof of the upper bound is from the trace-preserving constraint on the channel $\mathcal{R}^\text{opt}$, such that $\mathcal{B}\sim \set{\hat{B}_i}$ acts as a CPTP map on $\hat{P}_L$. See Ref.~\cite{10.1063/1.1459754} for a more detailed proof. Substituting Eq.~\eqref{eq:prof_neear_opti1} and Eq.~\eqref{eq:prof_neear_opti2} into Eq.~\eqref{eq:prof_neear_opti}, we have $\left(F^{\mathrm{opt}}\right)^2 \leq F\left(\mathcal{R}^{\mathrm{TC}} \circ \mathcal{N} \circ \mathcal{E}\right)\leq F^{\mathrm{opt}}$. As a straightforward corollary, we have 
    \begin{align}
            \frac{1}{2}\left(1 - F(\mathcal{R}^{\text{TC}}\circ\mathcal{N} \circ \mathcal{E}) \right)\leq 1 - F^{\text{opt}} \leq 1 -  F(\mathcal{R}^{\text{TC}}\circ\mathcal{N} \circ \mathcal{E}).
    \end{align}
    and conclude the proof.
\end{proof}

\begin{figure}[t]
    \centering
    \includegraphics[width = 0.45 \textwidth]{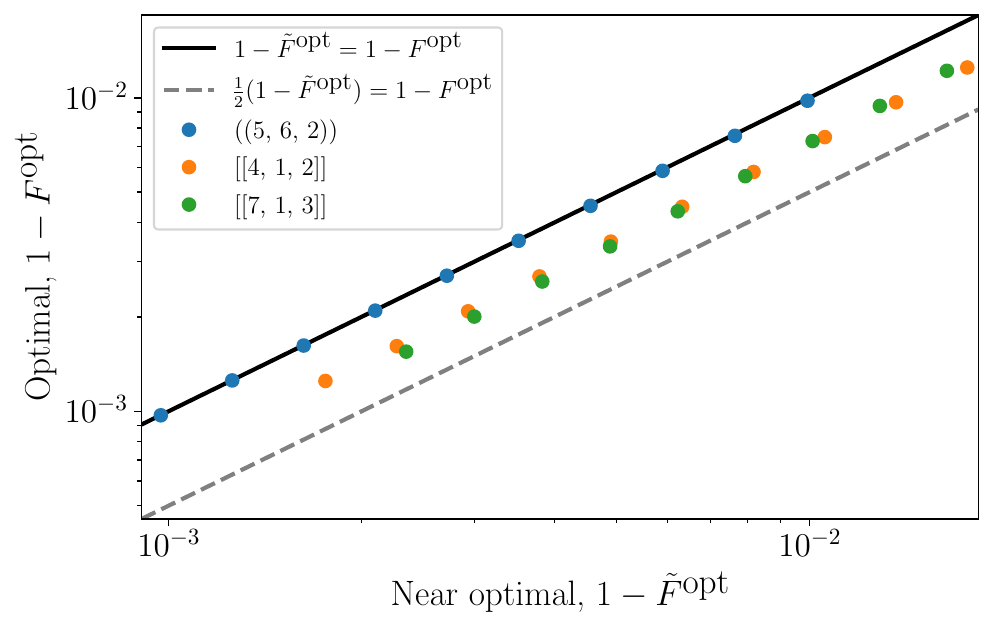}
    % \captionsetup{justification=justified}
     \caption{Comparison of near-optimal and optimal infidelity for qubit codes [[4, 1, 2]], ((5, 6, 2)), and [[7, 1, 3]] under amplitude damping noise. The black solid line (gray dashed line) represents the upper (lower) bound on the optimal infidelity given by the near-optimal infidelity}
    \label{fig:verification_qubit_supp}
\end{figure}

Since the goal is to find and benchmark good codes, we are mostly only interested in codes that perform well against the noise channels of interest. In these scenarios, the bound given by the TC channel is quite a tight bound on the optimal recovery. Therefore, the transpose channel provides a constructive channel that solves the slightly relaxed convex optimization problem. To numerically demonstrate the two-sided bound, we can look at some example qubit codes in Fig.~\ref{fig:verification_qubit_supp}. It is clear that while the two-sided bounds are narrow, they nicely bound the optimal code performances.

As a high-level understanding, the TC channel is a generalized version of the QEC recovery procedure from KL conditions: $\mathcal{A}\left(\hat{P}\right)^{-\frac{1}{2}}$ performs measurements that project the state into carefully chosen subspaces, and $\hat{P} \hat{N}_i^\dagger$ rotates them back into the corresponding logical space. More rigorously, we can derive the following lemma.

\begin{lemma}\label{lem:kraus}
    Given noise channel $\mathcal{N}\sim\set{\hat{N}_i}$ and logical codewords $\ket{\mu_L}$, the form of the transpose channel Kraus operators is equivalent to
    \begin{eqnarray}\label{eq:TC_kraus}
        \hat{R}_l^{\text{TC}} = \sum_{\mu, \nu, k}\left(M^{-\frac{1}{2}}\right)_{[\mu l], [\nu k]}|\mu_L\rangle\left\langle\nu_L\right| \hat{N}_{k}^{\dagger}
    \end{eqnarray}
    where the inverse represents matrix pseudoinverse and the QEC matrix $M_{[\mu l],[\nu k]}:=\left\langle\mu_L\left|\hat{N}_l^{\dagger} \hat{N}_k\right| \nu_L\right\rangle$.
\end{lemma}
\begin{proof}
    The error subspace basis is defined by the set $\set{\hat{N}_l\ket{\mu_L}}$. It is worth noting that this basis set is neither necessarily linear independent nor necessarily complete on the Hilbert space. However, it provides us the sufficient basis to express $\mathcal{N}\left(\hat{P}_L\right)^{-\frac{1}{2}}$ and consequently $\hat{R}_l^{\text{TC}}$. To start with, we can write a projection operator on this subspace
    \begin{eqnarray}
        \hat{P}_{N} = \sum_{l, k, \mu, \nu}\left(M^{-1}\right)_{[\mu l], [\nu k]} \hat{N}_l|\mu\rangle\left\langle\nu\right| \hat{N}_k^{\dagger},
    \end{eqnarray}
    where $M^{-1}$ is the pseudoinverse of $M$. The projector expression can be verified through checking its effect in the error subspace $\hat{P}_{N}\hat{N}_k| \nu_L\rangle = \hat{N}_k| \nu_L\rangle$.

    Notice that $\mathcal{N}\left(\hat{P}_L\right) = \sum_{l,\mu} \hat{N}_l \ket{\mu_L}\bra{\mu_L} \hat{N}_l^\dagger$. We can then derive representation of $\mathcal{N}\left(\hat{P}_L\right)^{-\frac{1}{2}}$ in the error subspace basis as
    \begin{eqnarray}
    \mathcal{N}\left(\hat{P}_L\right)^{-\frac{1}{2}}=\sum_{l, k, \mu, \nu}\left(M^{-\frac{3}{2}}\right)_{[\mu l], [\nu k]} \hat{N}_l|\mu_L\rangle\left\langle\nu_L\right| \hat{N}_{k}^{\dagger}.\label{eq:expr_inv_sqrt}
    \end{eqnarray}
    We can either derive step-by-step or straightforwardly verify from the definition of $\mathcal{N}\left(\hat{P}_L\right)^{-\frac{1}{2}}$:
    \begin{align}
        \mathcal{N}\left(\hat{P}_L\right)^{-\frac{1}{2}} \mathcal{N}\left(\hat{P}_L\right) \mathcal{N}\left(\hat{P}_L\right)^{-\frac{1}{2}} = \hat{P}_{N}. \label{eq:def_inv_sqrt}
    \end{align}
    With the expression presented in Eq.~\eqref{eq:expr_inv_sqrt}, the transpose channel Kraus can be written in the form of
    \begin{eqnarray}
        \hat{R}^{\text{TC}}_l &:=& \hat{P}_L \hat{N}_l^\dagger \mathcal{N}\left(\hat{P}_L\right)^{-\frac{1}{2}}\\
        &=& \sum_\mu \ket{\mu_L}\bra{\mu_L} \hat{N}_l^\dagger \sum_{p, k, \eta, \nu}\left(M^{-\frac{3}{2}}\right)_{[\eta p], [\nu k]} \hat{N}_p|\eta_L\rangle\left\langle\nu_L\right| \hat{N}_{k}^{\dagger}\\
        &=&  \sum_\mu \sum_{p, k, \eta, \nu}  \left(M^{-\frac{3}{2}}\right)_{[\eta p], [\nu k]} \bra{\mu_L} \hat{N}_l^\dagger\hat{N}_p|\eta_L\rangle \ket{\mu_L}\left\langle\nu_L\right| \hat{N}_{k}^{\dagger} \\
        &=&  \sum_\mu \sum_{p, k, \eta, \nu}  \left(M^{-\frac{3}{2}}\right)_{[\eta p], [\nu k]} M_{[\mu l], [\eta p]} \ket{\mu_L}\left\langle\nu_L\right| \hat{N}_{k}^{\dagger} \\
        &=& \sum_{\mu, \nu, k}\left(M^{-\frac{1}{2}}\right)_{[\mu l], [\nu k]}|\mu_L\rangle\left\langle\nu_L\right| \hat{N}_{k}^{\dagger}.
    \end{eqnarray}
\end{proof}

Although we assume isometry encoding for simplicity, it is straightforward to extend our result to the context of general channel reversal problems. For example, for mixed-state encodings, we can combine the encoding channel with the noise channel and revert the composed channel. For other codes like subsystem codes, the transpose channel is also valid following, for example, Ref.~\cite{PhysRevA.86.012335}.

While Eq.~\eqref{eq:TC_kraus} is useful for our proof of the near-optimal fidelity expression, it is not so intuitive. We can arrive at a more intuitive expression by defining a set of new basis 
\begin{eqnarray}
    \ket{\psi_{i}} := \sum_{[\nu k]}\left(M^{-\frac{1}{2}}\right)^\ast_{i, [\nu k]} \hat{N}_k\ket{\nu_L},
\end{eqnarray}
which is not necessarily a set of orthonormal basis. In fact, since $M$ is a Hermitian matrix, we can check that $\langle \phi_j\vert \phi_i\rangle = M^{-\frac{1}{2}}_{j, [\mu l]}\left(M^{-\frac{1}{2}}\right)^\ast_{i, [\nu k]} \bra{\mu_{L}}\hat{N}_l^\dagger \hat{N}_k\ket{\nu_{L}} = \left(M^{-1/2}M M^{-1/2}\right)_{j,i}$, which only guarantees orthonormality of the basis if $M$ is full rank. The transpose channel Kraus operators can be conveniently expressed as 
\begin{eqnarray}
    \hat{R}_l^{\text{TC}} = \sum_\mu \ket{\mu_L}\langle\psi_{[\mu l]}\vert.
\end{eqnarray}
This form of the transpose channel Kraus operators can be interpreted as following: for each logical codeword $\ket{\mu_L}$ and noise Kraus operator with label $l$, the physical operation of the transpose channel is to measure the projector $\vert \psi_{[\mu l]} \rangle \langle \psi_{[\mu l]}\vert$ and perform a rotation back to the logical codeword $\ket{\mu_L}$ accordingly.

\section{Derivation of the near-optimal fidelity expression}

With Lemma.~\ref{lem:kraus}, we arrive at our main result, Theorem~\ref{theo:Fe} in the main text.

\begin{theorem*}
    For a $d_L$-dimensional encoding, $\mathcal{E}$, and noise channel $\mathcal{N}$, the near-optimal fidelity is defined as
    \begin{align}\label{supp_eq:TC_main_result}
        \tilde{F}^{\text{opt}}=\frac{1}{d_L^2}\left\|\operatorname{Tr}_L \sqrt{M}\right\|_F^2,
    \end{align}
    where $M$ is the QEC matrix, $\left(\operatorname{Tr}_L B\right)_{l,k} = \sum_\mu B_{[\mu l], [\mu k]}$ denotes the partial trace over the code space indices, and $||\cdot||_F$ is the Frobenius norm. The near-optimal fidelity gives a two-sided bound on the optimal fidelity as
    \begin{eqnarray}
    \frac{1}{2}\left(1 - \tilde{F}^{\text{opt}} \right)\leq 1 - F^{\text{opt}} \leq 1 - \tilde{F}^{\text{opt}}.
    \end{eqnarray}
\end{theorem*}
\begin{proof}
    The near-optimal fidelity is defined as the fidelity achieved by the transpose channel recovery, $\tilde{F}^{\text{opt}}:= F\left(\mathcal{R}^{\text{TC}}\circ \mathcal{N}\circ \mathcal{E}\right)$. Substituting Lemma \ref{lem:kraus} into Eq.~\eqref{eq:ent_fid_expr}, we have
\begin{eqnarray}
    \tilde{F}^{\text{opt}} &=& \frac{1}{d_L^2} \sum_{i, j}\left|\operatorname{Tr}\left(\hat{R}^{\text{TC}}_i \hat{N}_j \right)\right|^2 \\
    &=& \frac{1}{d_L^2} \sum_{i, j}\left|\operatorname{Tr}\left(\sum_{\mu, \nu, k}\left(M^{-\frac{1}{2}}\right)_{[\mu i], [\nu k]}|\mu_L\rangle\left\langle\nu_L\right| \hat{N}_{k}^{\dagger}\hat{N}_j \right)\right|^2\\
    &=& \frac{1}{d_L^2} \sum_{i, j}\left|\operatorname{Tr}\left(\sum_{\mu, \nu, k}\left(M^{-\frac{1}{2}}\right)_{[\mu i], [\nu k]}\left\langle\nu_L\right| \hat{N}_{k}^{\dagger}\hat{N}_j |\mu_L\rangle \right)\right|^2\\
    &=& \frac{1}{d_L^2} \sum_{i, j}\left|\operatorname{Tr}\left(\sum_{\mu, \nu, k}\left(M^{-\frac{1}{2}}\right)_{[\mu i], [\nu k]}M_{[\nu k], [\mu j]} \right)\right|^2\\
    &=& \frac{1}{d_L^2} \sum_{i, j}\left|\sum_{\mu}\left(M^{\frac{1}{2}}\right)_{[\mu i], [\mu j]}\right|^2\\
    &=& \frac{1}{d_L^2} \left\| \text{Tr}_{L} \sqrt{M}\right\|_F^2
\end{eqnarray}
where we have used the cyclicity of trace and the definition of the Frobenius norm $||B||_F:= \sum_{i,j} \abs{B_{i,j}}^2$

\end{proof}
In some cases, we might also be dealing with non-orthonormal codewords. As an example, finite-energy GKP codes generally have non-orthonormal codewords. When we want to obtain a valid near-optimal fidelity for the orthonormalized codewords, it is convenient to have the following corollary.

\begin{corollary}
    For any non-orthonormal codeword with overlap matrix
    \begin{align}
        m_{\mu,\nu} := \bra{\mu_L}\ket{\nu_L},
    \end{align}
    the near-optimal fidelity in Theorem~\ref{theo:Fe} has the form of
    \begin{align}
        \tilde{F}^{\text{opt}}=\frac{1}{d_L^2}\left\|\operatorname{Tr}_L \sqrt{(m^{-1}\otimes I_{l})M}\right\|_F^2, \label{eq:Fe_nonorthonormal}
    \end{align}
    where $I_l$ is a $l\times l$ identity matrix, and $l$ is the number of Kraus operators of the noise channel $\mathcal{N}$.
\end{corollary}
\begin{proof}
Assume that with the, not necessarily orthonormal, original codewords $\ket{\mu_L}$, we obtain the QEC matrix $M$. Moreover, we can always find a matrix $A$, such that $A A^\dagger = m^{-1}\otimes I_l$. Consequently, $A^\dagger (m\otimes I_{l}) A = I$ and  orthonormalizes the codewords. The QEC matrix obtained from the orthonormalized codewords, $M^\prime$, is related to the original QEC matrix via $M^\prime = A^\dagger M A$. Therefore, based on Theorem \ref{theo:Fe}, we have
    \begin{align}
\tilde{F}^{\text{opt}}& =\frac{1}{d_L^2}\left\|\operatorname{Tr}_L \sqrt{M^\prime}\right\|_F^2. \label{eq:prof_col2_1}
\end{align}
Because of the tensor product structure of $m\otimes I_l$, we can write $A = B\otimes I_l$ and still utilize the cyclicity of the partial trace:
\begin{align}
    \operatorname{Tr}_L \sqrt{M^\prime} = \operatorname{Tr}_L A\sqrt{M^\prime}A^{-1} = \operatorname{Tr}_L \sqrt{AM^\prime A^{-1}} = \operatorname{Tr}_L \sqrt{AA^{\dagger}M} = \operatorname{Tr}_L \sqrt{(m^{-1}\otimes I_{l})M},\label{eq:prof_col2_2}
\end{align}
where the second equality is a consequence of the identity $\left(A\sqrt{M^\prime}A^{-1}\right)^2 = AM^\prime A^{-1}$. Substituting Eq.~\eqref{eq:prof_col2_2} into Eq.~\eqref{eq:prof_col2_1}, we conclude the proof.

\end{proof}

\section{Derivation of the perturbative form of the near-optimal fidelity and its properties}
Here we attempt to prove Corollary.~\ref{coro:Fe_pert}, which is a perturbative expansion of Theorem~\ref{theo:Fe}. Such a perturbative expansion can lead to fruitful analytical results, as we have demonstrated for the thermodynamic code and the GKP code in the main text.

\begin{corollary*}
     With a suitable choice of unitary gauge, there exists a decomposition of the QEC matrix $M$ such that $D = \frac{1}{d_L}\operatorname{Tr}_L M$ is diagonal. With the residual matrix $\Delta M := M - I_L \otimes D$, the near-optimal infidelity has a perturbative expansion through
    \begin{eqnarray}
    1 - \tilde{F}^{\text{opt}} =  \frac{1}{d_L}\norm{f(D)\odot \Delta M }_F^2 + \mathcal{O}\left(\frac{1}{d_L}\norm{f(D)\odot \Delta M }_F^3\right)
    \end{eqnarray}
    where we define $f(D) _{[\mu l],[\nu k]} = \frac{1}{\sqrt{D_{ll}} + \sqrt{D_{kk}}}$ and the Hadamard product to be $\left(A\odot B\right)_{ij} = A_{ij}B_{ij}$.
\end{corollary*}

\begin{proof}
Given the decomposition $M = I_L\otimes D + \Delta M$, where matrix $D$ is diagonal and $D\succ 0$ (since we can truncate the $M$ matrix and discard the channel subspaces which vanish,  $D_{ii}=0$). If we treat $\Delta M$ as a perturbation, we can perturbatively expand the matrix square root such that 
\begin{eqnarray}
    \sqrt{M} &=& I_L \otimes \sqrt{D} + f(D)\odot \Delta M - f(D)\odot \left(f(D)\odot \Delta M\right)^2 + \delta,\\
    ||\delta||_F &\le& \mathcal{O}(\norm{f(D)\odot \Delta M }_F^3),
\end{eqnarray}
which is a second-order expansion from the Daleckii-Krein theorem~\cite{DaletskiiKrein1965, carlsson2018perturbation}. Notice that $\text{Tr}D = 1$ and $\text{Tr}_L\Delta M = 0$, then we have
\begin{align}
    \text{Tr}_L\sqrt{M}=d_L \sqrt{D}-f(D) \odot \text{Tr}_L \left((f(D) \odot \Delta M)^2\right)+\text{Tr}_L \delta.
\end{align}
Therefore, an expansion of the Frobenius norm lead to the final expression
\begin{align}
      1-\tilde{F}^{\text{opt}} &= 1-\frac{1}{d_L^2}\left\|\operatorname{Tr}_L \sqrt{M}\right\|_F^2\\
      &= 1-\frac{1}{d_L^2} \operatorname{Tr} (\operatorname{Tr}_L \sqrt{M})^2 \\
    & = 1-\text{Tr} D + \Tr\left(\frac{1}{d_L}\Tr_L \left(f(D)\odot \Delta M\right)^2\right)+ \mathcal{O}\left(\frac{1}{d_L}\norm{f(D)\odot \Delta M }_F^3\right)\\
    & = \frac{1}{d_L}\norm{f(D)\odot \Delta M}_F^2+ \mathcal{O}\left(\frac{1}{d_L}\norm{f(D)\odot \Delta M }_F^3\right).
\end{align}

\end{proof}

As mentioned in the main text, we can choose to determine $D$ not through performing a unitary rotation but instead truncating the off-diagonal terms in $\frac{1}{d_L}\operatorname{Tr}_L M$. Such a choice can be more convenient when the analytical expression for the unitary rotation is unknown. In this scenario, we end up with $\text{Tr}\Delta M \neq 0$ and $\Tr D\neq 1$. Physically, we are distributing some of the correctable parts into the uncorrectable off-diagonal part of the QEC matrix, thus overestimating the near-optimal channel infidelity. The perturbative form is modified to be
\begin{align}
     1-\tilde{F}^{\text{opt}} =\frac{1}{d_L}\norm{f(D)\odot \Delta M}_F^2 - \frac{1}{d_L^2}\Tr\left(\Tr_L f(D)\odot \Delta M\right)^2+ \mathcal{O}\left(\frac{1}{d_L}\norm{f(D)\odot \Delta M }_F^3\right).
\end{align}
Therefore, with such a choice of $D$, the resulting infidelity is a good approximation as long as $\norm{f(D)\odot \left(\frac{1}{d_L} \operatorname{Tr}_L M - D\right)}_F^2\ll \frac{1}{d_L}\norm{ f(D)\odot \Delta M}_F^2$. Moreover, if higher precision is required, it is certainly possible to have higher orders of perturbative expressions through the definition of Fretchet derivatives or through solving the iterative relations order-by-order.

For a more intuitive understanding of the perturbative form, we can rewrite it as
\begin{eqnarray}
 1 - F^{\text{opt}} = \frac{1}{d_L} \left(\sum_{[\mu l]\neq [\nu k]} \frac{p_{[\mu l]} p_{[\nu k]}}{\left(\sqrt{\bar{p}_{l}} + \sqrt{\bar{p}_{k}}\right)^2}\Tr{\hat{P}_{[\mu l]}\hat{P}_{[\nu k]}}  + \sum_{[\mu l]}\frac{\left(p_{[\mu l]} - \bar{p}_{l}\right)^2}{4\bar{p}_{l}}\right) + \mathcal{O}\left(\frac{1}{d_L}\norm{f(D)\odot \Delta M }_F^3\right).
\end{eqnarray}
Here, we define the error subspace projectors as $\hat{P}_{[\mu l]} := \frac{\hat{N}_l\ket{\mu_L}\bra{\mu_L}\hat{N}_l^\dagger}{p_{[\mu l]}}$, where $p_{[\mu l]} := \bra{\mu_{L}}\hat{N}_l^\dagger\hat{N}_l\ket{\mu_{L}}$ and represents the probability of measuring $\ket{\mu, l}$ if we initialize in $\ket{\mu_{L}}$. The probability of measuring noise Kraus operator $l$ is $\bar{p}_{l} = \frac{1}{d_L}\sum_{\mu=0}^{d_L-1}p_{[\mu l]}$. In the expression, the second term contains the errors from different logical codewords reacting differently to the same noise Kraus operator, which leaks information to the environment. The first term includes logical flip errors (reverting to a different logical state) and dephasing errors (reverting to the same logical state but with an incorrect phase). 

Moreover, the perturbative form possesses a special property when it vanishes:
\begin{corollary}
      As $\frac{1}{d_L}\norm{f(D)\odot \Delta M}_F^2 \to 0$, the near-optimal fidelity
    \begin{align}
         \tilde{F}^{\text{opt}} \to \Tr{D}.
    \end{align}
\end{corollary}
\begin{proof}
    Here, we define $\Delta := \frac{1}{d_L}\abs{\norm{\Tr_L \sqrt{M}}_F - \norm{\Tr_L \sqrt{I_\mu \otimes D}}_F}$, and we show that $\Delta$ vanishes when the perturbative expression vanishes, i.e. $\frac{1}{d_L}\norm{f(D)\odot \Delta M}_F^2 \to 0$. First, notice some matrix identities
    \begin{eqnarray}
        \Delta &\leq& \frac{1}{d_L}\norm{\Tr_L \sqrt{M} - \Tr_L \sqrt{I_L \otimes D}}_F\\
        &\leq& \frac{1}{\sqrt{d_L}}\norm{\sqrt{M} - \sqrt{I_L \otimes D}}_F\\
        &\leq& \norm{\sqrt{\frac{M}{d_L}} - \sqrt{\frac{I_L \otimes D}{d_L}}}_F
    \end{eqnarray}
    The first inequality applies Von Neumann's trace inequality to obtain that for any two Hermitian matrices $A,B\succeq 0$, $\norm{A-B}_F \geq \norm{A}_F - \norm{B}_F$. In the second inequality, we proceeded by decomposing any matrix $X = \sum_i Y_i \otimes \sigma_i$, where $\sigma_i$ are the Gell-mann matrices, and show that $\norm{X}_F \geq \frac{1}{\sqrt{d_L}}\norm{\Tr_L X}_F$, where $d$ is the logical dimension. Notice that from the definition of Fretchet derivatives and the Daleckii-Krein theorem, 
    \begin{eqnarray}
        \sqrt{M} = \sqrt{I_L \otimes D} + f(D)\odot \Delta M + o(\Delta M).
    \end{eqnarray}
    Since $D\succ 0$ and $\Tr{D}\leq 1$, $\norm{f(D)\odot \frac{\Delta M}{\sqrt{d_L}}}_F \to 0$ implies that $\norm{\frac{\Delta M}{\sqrt{d_L}}}_F\to 0$. Thus,
        \begin{eqnarray}
        \lim_{\frac{\Delta M}{\sqrt{d_L}} \to 0} \Delta &\leq& \lim_{\frac{\Delta M}{\sqrt{d}_L} \to 0}\norm{f(\frac{D}{d_L})\odot\frac{\Delta M}{d_L} + o(\frac{\Delta M}{d_L})}_F\\
        &=& \norm{f(D)\odot\frac{\Delta M}{\sqrt{d_L}}}_F\\
        &=& 0
    \end{eqnarray}
    where the last equality holds by assumption. Therefore, in the limit of $\frac{1}{d_L}\norm{f(D)\odot \Delta M}_F^2 \to 0$,
    \begin{align}
     \tilde{F}^{\text{opt}} = \frac{1}{d_L^2}\left\|\operatorname{Tr}_L \sqrt{M}\right\|_F^2 =\frac{1}{d_L^2}\left\|\operatorname{Tr}_L \sqrt{I_L \otimes D}\right\|_F^2 = \Tr D.
\end{align}
\end{proof}

A special case is when $D$ contains all the correctable elements, i.e. $\text{Tr}D = 1$, where both the near-optimal and the optimal fidelity will be unity.

\section{Definitions and notations for examples of applications}

\subsection{Noise models}

Here we define the noise models that appeared in the examples we provided in the main text. The single qubit amplitude damping noise is defined as $\mathcal{N}(\hat{\rho}) = \hat{K}_0 \hat{\rho} \hat{K}_0^\dagger + \hat{K}_1 \hat{\rho} \hat{K}_1^\dagger$, where
\begin{equation}
    \hat{K}_0 = \ket{0}\bra{0} + \sqrt{1-p}\ket{0}\bra{0}, \hat{K}_1 = \sqrt{p}\ket{0}\bra{1},
\end{equation}
with $p$ being the damping probability. The amplitude damping noise is a non-Pauli noise, and it is unclear how to optimally decode it for, for example, stabilizer codes. Therefore, a common choice is to perform Pauli Twirling to convert it into a Pauli channel. A channel is a Pauli channel if it has the form of 
\begin{equation}
    \mathcal{N}(\hat{\rho}) = \left(1 - p_X- p_Y-p_Z\right) \hat{\rho} + p_X \hat{X}\hat{\rho}\hat{X}+ p_Y \hat{Y}\hat{\rho} \hat{Y}+ p_Z \hat{Z}\hat{\rho} \hat{Z}.
\end{equation}

For codes encoded in an oscillator, the excitation loss noise channel, also known as the pure loss channel, has the form of $\mathcal{N}(\hat{\rho}) = \sum_{i=0}^\infty \hat{N}_l \hat{\rho} \hat{N}_l^\dagger$, where
\begin{equation}
    \hat{N}_l = \left(\frac{\gamma}{1-\gamma}\right)^{l/2}\frac{\hat{a}^l}{\sqrt{l!}}\left(1-\gamma\right)^{\hat{n}/2},
\end{equation}
and $\hat{n} = \hat{a}^\dagger \hat{a}$ is the number operator. Here, $\gamma$ is the loss parameter.

\subsection{Code definitions}

For the standard qubit codes, we refer interested readers to the original references~\cite{1996Steane, PhysRevA.56.2567, PhysRevA.56.2567, PhysRevLett.79.953}. We focus on the definitions for the bosonic codes, which largely follow the definitions in Ref.~\cite{PhysRevA.97.032346}.

Cat codes are characterized by the coherent state amplitude, $\alpha$, and the loss spacing, $S$, which is also known as the number of legs. The codewords are defined as 
% \begin{equation}
% \ket{0_L}\propto \hat{\Pi}_0\ket{\alpha}, \ket{1_L}\propto \hat{\Pi}_{S+1}\ket{\alpha},
% \end{equation}
\begin{equation}
    \begin{aligned}
        \ket{0_L} &\propto \hat{\Pi}_0\ket{\alpha} \\
        \ket{1_L}& \propto \hat{\Pi}_{S+1}\ket{\alpha},
    \end{aligned}
\end{equation}
where we omit the normalization factors for simplicity. Here, $\hat{\Pi}_0 = \sum_{n=0}^\infty \vert 2n(S+1)\rangle \langle 2n(S+1)\vert$ and $\hat{\Pi}_{S+1} = \sum_{n=0}^\infty \vert (2n+1)(S+1)\rangle \langle (2n+1)(S+1)\vert$ are projectors onto Fock states that are $0$ and $S+1$ mod $2(S+1)$ respectively. In the special case of $S=0$, they are the even and odd parity space projectors. Notice that since the codewords' support in the Fock basis has a spacing of $S+1$, the error subspace is orthogonal and exactly correctable as long as we consider less than $S+1$ excitation loss. Nevertheless, it is unnecessary to have a complete separation in terms of their Fock state support: two states supported in the same generalized parity space can still be orthogonal. For example, the recently discovered squeezed cat code.~\cite{XuZheng2023} explores such a property. The squeezed cat code has its codewords living in the even and odd parity spaces, i.e. $S=0$ in the convention of cat codes. However, in the limit of infinite squeezing, the code can still correct single excitation loss because of the orthogonality mentioned.

Binomial codes are characterized by the loss spacing, $S$, and the dephasing spacing, $N$. The codewords are defined as 
\begin{equation}
    \begin{aligned}
        \ket{0_L} &= \frac{1}{\sqrt{2^{N+1}}} \sum_{m=0}^{N+1}\sqrt{{N+1 \choose m}}\ket{(S+1)m} \\
        \ket{1_L}&= \frac{1}{\sqrt{2^{N+1}}}  \sum_{m=0}^{N+1}\left(-1\right)^m\sqrt{{N+1 \choose m}}\ket{(S+1)m}.
    \end{aligned}
\end{equation}
It is clear that if we switch the basis to $\ket{\pm_L}$, the binomial code is protected by the Fock state subspace spacing very similar to the cat code. Therefore, it possesses the same property that it can exactly correct less than $S+1$ excitation loss.

The ideal GKP, which has infinite energy is defined on top of a symplective lattice. In the main text, we gave the square lattice as an example. In the position basis, we have that 
\begin{equation}
    \begin{aligned}
        \ket{0_L} &\propto \sum_{n\in\mathbb{Z}}\ket{2n\sqrt{\pi}}_x \\
        \ket{1_L}&\propto \sum_{n\in\mathbb{Z}}\ket{(2n+1)\sqrt{\pi}}_x.
    \end{aligned}
\end{equation}
Since these codewords are not physical, we apply a Gaussian envelope. We define the finite-energy GKP codes as
\begin{equation}
    \ket{\mu^\Delta_L} \propto e^{-\Delta^2 \hat{n}} \ket{\mu_L}
\end{equation}
where we omit the normalization factor. For GKP codes encoding two logical dimensions, the codewords live in the even parity subspace. Therefore, they are exact codes against a single excitation loss, but for more loss, they are approximate codes. For general lattices and logical dimensions, the separation in parity does not necessarily hold.

In the main text, we considered the performance comparison of cat, binomial, and GKP codes with increasing photon number. In particular, for cat and binomial codes, we increase their photon number by fixing the parameter S and increasing the coherent state amplitude, $\alpha$, and the dephasing spacing, $N$, respectively. For the GKP code, we fix the underlying lattice geometry and decrease $\Delta$.

\section{Thermodynamic code and derivations}

\subsection{Code definition}

Thermodynamic codes were first proposed in Ref.~\cite{PhysRevLett.123.110502}, where it was observed that the Eigenstate Thermalization Hypothesis automatically yielded approximate error correction codes in the bulk of the spectrum. We focus here on the simplest thermodynamic code, which can be seen as either an ETH code with respect to the one local Hamiltonian $\sum_i (I - \sigma_z^i)$, or a ground state encoding in the ferromagnetic Heisenberg model. The code words on $N$ modes, for $m = 0, 1, ..., N$ are 
\begin{equation}
    | h_m^N \rangle = \genfrac(){0pt}{0}{N}{N/2 + m/2}^{-1/2} \sum_{ \mathbf{s}: \sum_j s_j = m } | \mathbf{s} \rangle_N, 
\end{equation}
where $s_j = \pm 1$.

We focus on a logical qubit with $k = 2$. To define a logical qubit, we pick some $d < N, m_0 < N - d$ and take \begin{equation}
    \begin{aligned}
        | 0_L \rangle &= | h_{m_0}^N \rangle, \\
        | 1_L \rangle &= | h_{m_0 + d}^N \rangle,
    \end{aligned}
\end{equation}
where $d$ is the distance of the code. In our example given in the main text, we focused on the case where $m_0 = \frac{d}{2}$.

\subsection{One erasure}

We first work out the case of a single erasure. The error channel can be written with Kraus operators 
\begin{equation}
    \begin{aligned}
        K_0 &= \sqrt{1 - p} I, \\
        K_1 &= \sqrt{p} | \Omega \rangle \langle -1 |, \\
        K_2 &= \sqrt{p} | \Omega \rangle \langle 1 |,
    \end{aligned}
\end{equation}
where $p$ is the probability of erasure, and $| \Omega \rangle$ is some external state, orthonormal to $|1\rangle, |-1\rangle$, representing an erased qubit. Since the codewords are permutation invariant, without loss of generality, we can set the error-prone qubit to be the first qubit. Then the above Kraus operators are only supported on the first qubit, with identity on the rest of the system. We want to calculate the relevant QEC matrix.

First, note that if we simply pick $d > 2$ we will have $M_{l, \mu, l' , \nu} \propto \delta_{\mu, \nu} \delta_{l, l'}$ such that 
\begin{equation}
    M = \sum_{l, \mu} c_{\mu, l} | l \rangle \langle l | \otimes | \mu \rangle \langle \mu |.
\end{equation}
Then, for this problem, Eq.~(\ref{supp_eq:TC_main_result}) simplifies to
\begin{equation}
    1 - \tilde{F}^{\text{opt}} = 1 - \frac{1}{4} \sum_l \left( \sum_{\mu} \sqrt{c_{\mu, l}} \right)^2.
\end{equation}

We simply have to calculate $c_{\mu, l}$, which are defined by \begin{equation}
    \langle \mu_L | K_l^{\dag} K_l | \mu_L \rangle.
\end{equation}

For $l = 0$ this is simply $1- p$. Suppose $\mu_L$ has magnetization $m$. Then $K_1$ increases the magnetization by $1$, so that this simply counts the number of states on $N-1$ qubits with magnetization $m+1$. Similarly, $K_2$ decreases the magnetization by $1$. Then we have \begin{equation}
    \begin{aligned}
        &\langle h_m^N | K_1^{\dag} K_1 | h_m^N \rangle = p \frac{{N - 1\choose (N + m)/2}}{{N \choose (N + m)/2}} = \frac{p}{2} (1 - m/N), \qquad 
        &\langle h_m^N | K_2^{\dag} K_2 | h_m^N \rangle = p \frac{{N - 1\choose (N + m-2)/2}}{{N \choose (N + m)/2}} = \frac{p}{2} (1 + m/N)
    \end{aligned}
\end{equation} 
Specializing to the chosen codewords, 
\begin{equation}
    \begin{aligned}
        \sum_l \left( \sum_{\mu} \sqrt{c_{\mu, l}} \right)^2  
    &= \sum_{l = 0}^2 \left( \sqrt{\langle h_{m_0}^N | K_l^{\dag} K_l | h_{m_0}^N \rangle} + \sqrt{\langle h_{m_0+d}^N | K_l^{\dag} K_l | h_{m_0+d}^N \rangle}  \right)^2 \\
    &= 4 (1-p) + \frac{p}{2}\left( \sqrt{ 1 - \frac{m_0}{N}} + \sqrt{1 - \frac{m_0 + d}{N}} \right)^2 + \frac{p}{2}\left( \sqrt{ 1 + \frac{m_0}{N}} + \sqrt{1 + \frac{m_0 + d}{N}} \right)^2 \\ 
    &= 4 (1-p) + 2 p + p \sqrt{ 1 - \frac{m_0}{N}} \sqrt{1 - \frac{m_0 + d}{N}} + p  \sqrt{ 1 + \frac{m_0}{N}} \sqrt{1 + \frac{m_0 + d}{N}} 
    \end{aligned}
\end{equation}

The most straightforward way to pick two states in the middle of the spectrum is to take $m_0 = -d/2$. We can treat $x \equiv d/N$ as a small parameter in the thermodynamic limit. The exact expression is \begin{equation}
    1- \tilde{F}^{\text{opt}} = \frac{1}{2} p(1 - \sqrt{1 - x^2/4}).
\end{equation}

For small $x$, we get \begin{equation}
    1 - \tilde{F}^{\text{opt}} = \frac{1}{16} p x^2 = \frac{1}{16} p \frac{d^2}{N^2}
\end{equation}
which has the expected $1/N^2$ scaling \cite{PhysRevX.10.041018}. At first glance, a strange property of this scaling is that infidelity increases with distance. We can understand this as follows: once $d > 2$, the QEC matrix becomes diagonal. Then, infidelity is given by how close the diagonal elements are to each other. The smaller $d$ is, the closer the codewords are to each other, such that the diagonal elements become more similar. Another way to think about this is that now $d$ controls how far from the exact middle of the spectrum we are.

\subsection{Two erasures}

We now fix two qubits $i, j$ where erasures occur. Again, without loss of generality, we can simply think of these as the first two qubits. If we pick $d > 4$, the QEC matrix will be diagonal along the codewords -- however, it will not be diagonal along the Kraus operators. We can index Kraus operators by $(\alpha, \beta)$, and write 
\begin{equation}
    \begin{aligned}
        M &= \sum_{\mu} \sum_{\alpha, \beta, \alpha', \beta'} c_{\mu, (\alpha, \beta), (\alpha', \beta')} | \alpha, \beta \rangle \langle \alpha', \beta' | \otimes | \mu \rangle \langle \mu | 
    \end{aligned}
\end{equation}
This is certainly diagonal for the case where there is only an error on one qubit, as we know from studying one erasure, so
\begin{equation}
    \begin{aligned}
        M = 
        &\sum_{\mu} c_{\mu, (0, 0), (0, 0)} | 0, 0 \rangle \langle 0, 0 | \otimes | \mu \rangle \langle \mu | 
        \\&+  
        \sum_{\mu} \sum_{\alpha \neq 0} c_{\mu, (\alpha, 0), (\alpha, 0)} \left( | \alpha, 0 \rangle \langle \alpha, 0 | + | 0, \alpha \rangle \langle 0, \alpha |\right)  \otimes | \mu \rangle \langle \mu | 
        \\&+
        \sum_{\mu} \sum_{\alpha \neq 0, \beta \neq 0, \alpha' \neq 0, \beta' \neq 0} c_{\mu, (\alpha, \beta), (\alpha', \beta')} | \alpha, \beta \rangle \langle \alpha', \beta' | \otimes | \mu \rangle \langle \mu |,
    \end{aligned}
\end{equation}
where the first row corresponds to no error, the second row to one error, and the last row to two errors. The one error elements can be read off using the results of the previous section: $c_{\mu, (0, 0), (0, 0)} = (1 - p)^2,
        c_{\mu, (1, 0), (1, 0)} = \frac{1}{2}p(1 - p) ( 1 - m_0/N),
        c_{\mu, (2, 0), (2, 0)} = \frac{1}{2}p(1 - p) ( 1 + m_0/N) $.

Now, let's consider the two error parts of the QEC matrix (i.e. both $\alpha, \beta$ non-zero). First, note that if $\alpha = \beta$, then we must have $\alpha' = \beta' = \alpha = \beta$ for the matrix element to be nonzero. The reason for this is that we are either taking away two excitations, or adding two excitations, so in order to have a non-zero inner product, we must do the same to the other side. With that, we have 
\begin{equation}
    \begin{aligned}
        M_{\rm 2 err}^{(\mu)}
        &= \begin{pmatrix}
            c_{\mu,(1, 1), (1, 1)} & 0 \\
            0 & c_{\mu,(2, 2), (2, 2)}
        \end{pmatrix} \oplus \begin{pmatrix}
            c_{\mu,(1, 2), (1, 2)} & c_{\mu,(1, 2), (2, 1)} \\
             c_{\mu,(2, 1), (1, 2)}  & c_{\mu,(2, 1), (2, 1)} \\
        \end{pmatrix}
    \end{aligned},
\end{equation}
so thankfully, we only have the one $2 \times 2$ matrix to diagonalize. The diagonal part has elements,
\begin{equation}
    \begin{aligned}
        c_{\mu,(1, 1), (1, 1)}
        &= \langle h_m^N | K_1^{\otimes 2 \dag} K_1^{\otimes 2} | h_m^N \rangle 
        = p^2 {N \choose (N + m)/2}^{-1} {N - 2 \choose (N + m) / 2} 
        = p^2 \left( \frac{(N - m)^2}{4N(N-1)} - \frac{N-m}{2N(N-1)}\right),\\
        c_{\mu,(2, 2), (2, 2)} 
        &=\langle h_m^N | K_2^{\otimes 2 \dag} K_2^{\otimes 2} | h_m^N \rangle 
        = p^2 {N \choose (N + m)/2}^{-1} {N - 2 \choose (N + m  - 4) / 2}
        = p^2 \left( \frac{(N + m)^2}{4N(N-1)} - \frac{N+m}{2N(N-1)}\right).
    \end{aligned}
\end{equation}

Now for the non-diagonal part. It is straightforward to see that all the elements are the same, which we will denote \begin{equation}
    M_{12} = \begin{pmatrix}
        c_{\mu,(1, 2), (1, 2)} & c_{\mu,(1, 2), (2, 1)} \\
         c_{\mu,(2, 1), (1, 2)}  & c_{\mu,(2, 1), (2, 1)} \\
    \end{pmatrix} \equiv c_{\mu, \times} \begin{pmatrix}
        1 & 1 \\ 1 & 1
    \end{pmatrix} \rightarrow \sqrt{M_{12}} = \sqrt{\frac{c_{\mu, \times} }{2}}\begin{pmatrix}
        1 & 1 \\ 1 & 1
    \end{pmatrix} 
\end{equation}

Since errors of the form $(1, 2)$ do not change the magnetization, we can calculate $c_{\mu, \times}$ as follows: 
\begin{equation}
    \begin{aligned}
        c_{\mu, \times}
        &= \langle h_m^N | (K_1^{\dag}\otimes K_2^{\dag})(K_1 \otimes K_2) | h_m^N \rangle 
        = p^2 {N \choose (N + m)/2}^{-1} {N - 2 \choose (N + m - 2) / 2}  
        % = p^2 \frac{(N+m)(N+m-2)}{4N(N-1)}
        = p^2 \frac{N^2 - m^2}{4N(N-1)}
    \end{aligned}
\end{equation}

With that we can now calculate the infidelity. Denoting the diagonal and non diagonal parts of $M$ by $M_1, M_2$ so that $M = M_1 \oplus M_2 $ and 
\begin{equation}
    \begin{aligned}
        M_1 = 
        &\sum_{\mu} c_{\mu, (0, 0), (0, 0)} | 0, 0 \rangle \langle 0, 0 | \otimes | \mu \rangle \langle \mu | 
        \\&+  
        \sum_{\mu} \sum_{\alpha \neq 0} c_{\mu, (\alpha, 0), (\alpha, 0)} \left( | \alpha, 0 \rangle \langle \alpha, 0 | + | 0, \alpha \rangle \langle 0, \alpha |\right)  \otimes | \mu \rangle \langle \mu | 
        \\&+
        \sum_{\mu} \sum_{\alpha = \beta \neq 0} c_{\mu, (\alpha, \alpha), (\alpha, \alpha)} | \alpha, \alpha \rangle \langle \alpha, \alpha | \otimes | \mu \rangle \langle \mu |, \\
        M_2 = &\oplus_{\mu} c_{\mu, \times} \begin{pmatrix}
            1 & 1 \\ 1 & 1
        \end{pmatrix},
    \end{aligned}
\end{equation}
one can plug and chug to obtain,
\begin{equation}
    \begin{aligned}
        ({\rm Tr}_{\mu} \sqrt{M})^2 = & ({\rm Tr}_{\mu} \sqrt{M_1})^2 + ({\rm Tr}_{\mu} \sqrt{M_2})^2 \\
        \| ({\rm Tr}_{\mu} \sqrt{M_1}) \|^2 = 
        & \left( \sum_{\mu} \sqrt{c_{\mu, (0, 0), (0, 0)}} \right)^2
        + 2 \sum_{\alpha = 1, 2} \left( \sum_{\mu} \sqrt{c_{\mu, (\alpha, 0), (\alpha, 0)}} \right)^2 + \sum_{\alpha = 1, 2} \left( \sum_{\mu} \sqrt{c_{\mu, (\alpha, \alpha), (\alpha, \alpha)}} \right)^2 \\
        \| {\rm Tr}_{\mu} \sqrt{M_2}\|^2 = &2 \left( \sum_{\mu} \sqrt{c_{\mu,\times}} \right)^2
    \end{aligned}
\end{equation}

Finally, the fidelity will be 
\begin{equation}
    \tilde{F}^{\text{opt}} = (1 - p)^2 + \frac{1}{2} \sum_{\alpha = 1, 2} \left( \sum_{\mu} \sqrt{c_{\mu, (\alpha, 0), (\alpha, 0)}} \right)^2 + \frac{1}{4}  \sum_{\alpha = 1, 2} \left( \sum_{\mu} \sqrt{c_{\mu, (\alpha, \alpha), (\alpha, \alpha)}} \right)^2  + \frac{1}{2} \left( \sum_{\mu} \sqrt{c_{\mu,\times}} \right)^2.
\end{equation}

We can now set $m_0 = - d/2, m_1 = d/2$ as usual, and obtain exact expressions. The full expression however, is long and tedious, so we take the limit $N \gg 1, d/N \ll 1$ in turn, and obtain 
\begin{equation}
    1 - \tilde{F}^{\text{opt}} = \left( \frac{1}{2} p - \frac{3}{8} p^2 \right) \frac{d^2}{N^2} + O \left( \frac{d^2}{N^3} \right)   .
\end{equation}
Note that in the $1$ erasure case, we could have taken $m \sim O(N)$ and still obtained the $d^2/N^2$ scaling. However, here we cannot -- if we take $m_0 = (N - d)/2$ for instance, it turns out there will remain uncorrectable errors even in the thermodynamic limit.

\subsection{Constant number of erasures under sufficiently large distance}

Now we consider the case where there are $l$ erasures for an instance of a thermodynamic code with distance $d > l$. Again, WLOG we place these errors on the first $l$ qubits. The matrix elements we have to calculate all have the form \begin{equation}\label{supp_eq:thermo_approx_form_constant_error}
\begin{aligned}
    c_{\mu}^{(l_1)}
    &= p^l (1 - p)^{N-l} \langle h_{m_{\mu}}^N |\left( K_1^{\otimes l_1} \otimes K_2^{\otimes l - l_1} \otimes K_0^{\otimes N - l} \right)^{\dag} \left( K_1^{\otimes l_1} \otimes K_2^{\otimes l - l_1} \otimes K_0^{\otimes N - l} \right) | h_{m_{\mu}}^N \rangle \\ 
    &= p^l (1 - p)^{N-l} {N \choose (N+m_{\mu})/2}^{-1} {N - l \choose ((N + m_{\mu})/2 - (l - l_1)},
\end{aligned}
\end{equation}
where the second factor counts the number of states on $N-l$ qubits with magnetization $m_{\mu} - l + 2 l_1$. Here, $m_{\mu}$ is the magnetization of the codeword $|\mu_L\rangle$. We will set $p = 1$, since we are interested in what happens if these $l$ qubits were deterministically erased. We will assume that $m_{\mu} \ll N$, in which case we can expand the above in powers of $m_{\mu}/N$ as follows:
\begin{equation}
\begin{aligned}
    c_{\mu}^{(l_1)}
    &\simeq \frac{1}{N^l} \left( \frac{N + m_{\mu}}{2} \right)^{l - l_1} \left( \frac{N - m_{\mu}}{2} \right)^{l_1} \\
    &= \frac{1}{2^{l}} \left(1 + \frac{m_{\mu}}{N} \right)^{l - l_1} \left(1 - \frac{m_{\mu}}{N} \right)^{ l_1} \\
    &\simeq \frac{1}{2^{l}} \left(1 + (l - 2l_1) \frac{m_{\mu}}{N} +  \frac{1}{2}(l^2 + 4 l_1^2 - 4 l l_1 - l) \frac{m_{\mu}^2}{N^2}\right) + O\left(\frac{1}{N^3} \right).
\end{aligned}
\end{equation}
where we have obtained the first approximate equality by only keeping track of terms with $N, m_{\mu}$.

Now, any permutation of $l_1$ $K_1$'s and $l - l_1$ $K_2$'s on the first $l$ qubits will give the same overlap (and any other configuration of Kraus operators will give $0$ overlap), so the above gives a block of the QEC matrix of size ${l \choose l_1}$ by ${l \choose l_1}$, with every matrix element being the same number above. Let us denote a matrix of ones of size $d \times d$ by $O_1^{d}$. The full QEC matrix then takes the form \begin{equation}
    M = \bigoplus_{\mu = 0, 1} \bigoplus_{l_1 = 0}^l c_{\mu}^{(l_1)} O_1^{{l \choose l_1}},
\end{equation}
such that \begin{equation}
\begin{aligned}
    \widetilde{F}^{\rm opt} 
    = \frac{1}{4} \| {\rm Tr}_{\mu} \sqrt{M} \|^2_F  
    &= \frac{1}{4} \sum_{l_1 = 0}^l {\rm Tr} \left( \sum_{\mu} \sqrt{c_{\mu}^{(l_1)} O_1^{{l \choose l_1}}}\right)^2 \\
    &= \frac{1}{4} \sum_{l_1 = 0}^l {l \choose l_1} \left( \sum_{\mu} \sqrt{c_{\mu}^{(l_1)}}\right)^2.
\end{aligned}
\end{equation}
The exact expression can now be easily evaluated numerically, or we can continue with our expansion in terms of $d/N$ as before. Setting $m_{\mu} = \pm d/2, x = d/N$, and using the approximate form in Eq.~\eqref{supp_eq:thermo_approx_form_constant_error}, we have
\begin{equation}
    \left( \sum_{\mu} \sqrt{c_{\mu}^{(l_1)}}\right)^2 \simeq \frac{4}{2^{l}} \left( 1 + \frac{1}{16}\left( -2 l + l^2 - 4 l l_1 + 4 l_1^2 \right) x^2 \right)
\end{equation}
To evaluate the prefactor of $x^2$, we note that $\sum_{l_1 = 0}^l {l \choose l_1} = 2^l, \sum_{l_1 = 0}^l {l \choose l_1} l_1 =  2^{l-1} l, \sum_{l_1 = 0}^l {l \choose l_1} l_1^2 = 2^{l-2}l(l-1)$. This leads us to the simple final expression, 
\begin{equation}\label{supp_eq:l_erasure_infid}
\begin{aligned}
    1 - \widetilde{F}^{\rm opt} 
    = \frac{l}{16} x^2 + O\left( \frac{1}{N^3} \right) = \frac{l}{16} \frac{d^2}{N^2} + O\left( \frac{1}{N^3} \right).
\end{aligned}
\end{equation}
Note that the above equation is only valid when we keep $l$ constant and take $N$ to be large, but will not hold if we allow $l$ to scale with $N$. Eq.~\eqref{supp_eq:l_erasure_infid} matches with the single and two erasure error expressions by setting $l=1,2$ and unit erasure probability $p=1$.


\begin{thebibliography}{78}%
    \makeatletter
    \providecommand \@ifxundefined [1]{%
     \@ifx{#1\undefined}
    }%
    \providecommand \@ifnum [1]{%
     \ifnum #1\expandafter \@firstoftwo
     \else \expandafter \@secondoftwo
     \fi
    }%
    \providecommand \@ifx [1]{%
     \ifx #1\expandafter \@firstoftwo
     \else \expandafter \@secondoftwo
     \fi
    }%
    \providecommand \natexlab [1]{#1}%
    \providecommand \enquote  [1]{``#1''}%
    \providecommand \bibnamefont  [1]{#1}%
    \providecommand \bibfnamefont [1]{#1}%
    \providecommand \citenamefont [1]{#1}%
    \providecommand \href@noop [0]{\@secondoftwo}%
    \providecommand \href [0]{\begingroup \@sanitize@url \@href}%
    \providecommand \@href[1]{\@@startlink{#1}\@@href}%
    \providecommand \@@href[1]{\endgroup#1\@@endlink}%
    \providecommand \@sanitize@url [0]{\catcode `\\12\catcode `\$12\catcode `\&12\catcode `\#12\catcode `\^12\catcode `\_12\catcode `\%12\relax}%
    \providecommand \@@startlink[1]{}%
    \providecommand \@@endlink[0]{}%
    \providecommand \url  [0]{\begingroup\@sanitize@url \@url }%
    \providecommand \@url [1]{\endgroup\@href {#1}{\urlprefix }}%
    \providecommand \urlprefix  [0]{URL }%
    \providecommand \Eprint [0]{\href }%
    \providecommand \doibase [0]{https://doi.org/}%
    \providecommand \selectlanguage [0]{\@gobble}%
    \providecommand \bibinfo  [0]{\@secondoftwo}%
    \providecommand \bibfield  [0]{\@secondoftwo}%
    \providecommand \translation [1]{[#1]}%
    \providecommand \BibitemOpen [0]{}%
    \providecommand \bibitemStop [0]{}%
    \providecommand \bibitemNoStop [0]{.\EOS\space}%
    \providecommand \EOS [0]{\spacefactor3000\relax}%
    \providecommand \BibitemShut  [1]{\csname bibitem#1\endcsname}%
    \let\auto@bib@innerbib\@empty
    %</preamble>
    \bibitem [{\citenamefont {Knill}\ and\ \citenamefont {Laflamme}(1997)}]{PhysRevA.55.900}%
      \BibitemOpen
      \bibfield  {author} {\bibinfo {author} {\bibfnamefont {E.}~\bibnamefont {Knill}}\ and\ \bibinfo {author} {\bibfnamefont {R.}~\bibnamefont {Laflamme}},\ }\href {https://doi.org/10.1103/PhysRevA.55.900} {\bibfield  {journal} {\bibinfo  {journal} {Phys. Rev. A}\ }\textbf {\bibinfo {volume} {55}},\ \bibinfo {pages} {900} (\bibinfo {year} {1997})}\BibitemShut {NoStop}%
    \bibitem [{\citenamefont {Bravyi}\ and\ \citenamefont {Kitaev}(1998)}]{bravyi1998quantum}%
      \BibitemOpen
      \bibfield  {author} {\bibinfo {author} {\bibfnamefont {S.~B.}\ \bibnamefont {Bravyi}}\ and\ \bibinfo {author} {\bibfnamefont {A.~Y.}\ \bibnamefont {Kitaev}},\ }\href@noop {} {\bibinfo {title} {Quantum codes on a lattice with boundary}} (\bibinfo {year} {1998}),\ \Eprint {https://arxiv.org/abs/quant-ph/9811052} {arXiv:quant-ph/9811052 [quant-ph]} \BibitemShut {NoStop}%
    \bibitem [{\citenamefont {Shor}(1995)}]{PhysRevA.52.R2493}%
      \BibitemOpen
      \bibfield  {author} {\bibinfo {author} {\bibfnamefont {P.~W.}\ \bibnamefont {Shor}},\ }\href {https://doi.org/10.1103/PhysRevA.52.R2493} {\bibfield  {journal} {\bibinfo  {journal} {Phys. Rev. A}\ }\textbf {\bibinfo {volume} {52}},\ \bibinfo {pages} {R2493} (\bibinfo {year} {1995})}\BibitemShut {NoStop}%
    \bibitem [{\citenamefont {Leung}\ \emph {et~al.}(1997)\citenamefont {Leung}, \citenamefont {Nielsen}, \citenamefont {Chuang},\ and\ \citenamefont {Yamamoto}}]{PhysRevA.56.2567}%
      \BibitemOpen
      \bibfield  {author} {\bibinfo {author} {\bibfnamefont {D.~W.}\ \bibnamefont {Leung}}, \bibinfo {author} {\bibfnamefont {M.~A.}\ \bibnamefont {Nielsen}}, \bibinfo {author} {\bibfnamefont {I.~L.}\ \bibnamefont {Chuang}},\ and\ \bibinfo {author} {\bibfnamefont {Y.}~\bibnamefont {Yamamoto}},\ }\href {https://doi.org/10.1103/PhysRevA.56.2567} {\bibfield  {journal} {\bibinfo  {journal} {Phys. Rev. A}\ }\textbf {\bibinfo {volume} {56}},\ \bibinfo {pages} {2567} (\bibinfo {year} {1997})}\BibitemShut {NoStop}%
    \bibitem [{\citenamefont {Gottesman}\ \emph {et~al.}(2001)\citenamefont {Gottesman}, \citenamefont {Kitaev},\ and\ \citenamefont {Preskill}}]{PhysRevA.64.012310}%
      \BibitemOpen
      \bibfield  {author} {\bibinfo {author} {\bibfnamefont {D.}~\bibnamefont {Gottesman}}, \bibinfo {author} {\bibfnamefont {A.}~\bibnamefont {Kitaev}},\ and\ \bibinfo {author} {\bibfnamefont {J.}~\bibnamefont {Preskill}},\ }\href {https://doi.org/10.1103/PhysRevA.64.012310} {\bibfield  {journal} {\bibinfo  {journal} {Phys. Rev. A}\ }\textbf {\bibinfo {volume} {64}},\ \bibinfo {pages} {012310} (\bibinfo {year} {2001})}\BibitemShut {NoStop}%
    \bibitem [{\citenamefont {Leghtas}\ \emph {et~al.}(2013)\citenamefont {Leghtas}, \citenamefont {Kirchmair}, \citenamefont {Vlastakis}, \citenamefont {Schoelkopf}, \citenamefont {Devoret},\ and\ \citenamefont {Mirrahimi}}]{PhysRevLett.111.120501}%
      \BibitemOpen
      \bibfield  {author} {\bibinfo {author} {\bibfnamefont {Z.}~\bibnamefont {Leghtas}}, \bibinfo {author} {\bibfnamefont {G.}~\bibnamefont {Kirchmair}}, \bibinfo {author} {\bibfnamefont {B.}~\bibnamefont {Vlastakis}}, \bibinfo {author} {\bibfnamefont {R.~J.}\ \bibnamefont {Schoelkopf}}, \bibinfo {author} {\bibfnamefont {M.~H.}\ \bibnamefont {Devoret}},\ and\ \bibinfo {author} {\bibfnamefont {M.}~\bibnamefont {Mirrahimi}},\ }\href {https://doi.org/10.1103/PhysRevLett.111.120501} {\bibfield  {journal} {\bibinfo  {journal} {Phys. Rev. Lett.}\ }\textbf {\bibinfo {volume} {111}},\ \bibinfo {pages} {120501} (\bibinfo {year} {2013})}\BibitemShut {NoStop}%
    \bibitem [{\citenamefont {Lescanne}\ \emph {et~al.}(2020)\citenamefont {Lescanne}, \citenamefont {Villiers}, \citenamefont {Peronnin}, \citenamefont {Sarlette}, \citenamefont {Delbecq}, \citenamefont {Huard}, \citenamefont {Kontos}, \citenamefont {Mirrahimi},\ and\ \citenamefont {Leghtas}}]{NatureLescanne2020}%
      \BibitemOpen
      \bibfield  {author} {\bibinfo {author} {\bibfnamefont {R.}~\bibnamefont {Lescanne}}, \bibinfo {author} {\bibfnamefont {M.}~\bibnamefont {Villiers}}, \bibinfo {author} {\bibfnamefont {T.}~\bibnamefont {Peronnin}}, \bibinfo {author} {\bibfnamefont {A.}~\bibnamefont {Sarlette}}, \bibinfo {author} {\bibfnamefont {M.}~\bibnamefont {Delbecq}}, \bibinfo {author} {\bibfnamefont {B.}~\bibnamefont {Huard}}, \bibinfo {author} {\bibfnamefont {T.}~\bibnamefont {Kontos}}, \bibinfo {author} {\bibfnamefont {M.}~\bibnamefont {Mirrahimi}},\ and\ \bibinfo {author} {\bibfnamefont {Z.}~\bibnamefont {Leghtas}},\ }\href {https://doi.org/10.1038/s41567-020-0824-x} {\bibfield  {journal} {\bibinfo  {journal} {Nature Physics}\ }\textbf {\bibinfo {volume} {16}},\ \bibinfo {pages} {509} (\bibinfo {year} {2020})}\BibitemShut {NoStop}%
    \bibitem [{\citenamefont {Michael}\ \emph {et~al.}(2016)\citenamefont {Michael}, \citenamefont {Silveri}, \citenamefont {Brierley}, \citenamefont {Albert}, \citenamefont {Salmilehto}, \citenamefont {Jiang},\ and\ \citenamefont {Girvin}}]{PhysRevX.6.031006}%
      \BibitemOpen
      \bibfield  {author} {\bibinfo {author} {\bibfnamefont {M.~H.}\ \bibnamefont {Michael}}, \bibinfo {author} {\bibfnamefont {M.}~\bibnamefont {Silveri}}, \bibinfo {author} {\bibfnamefont {R.~T.}\ \bibnamefont {Brierley}}, \bibinfo {author} {\bibfnamefont {V.~V.}\ \bibnamefont {Albert}}, \bibinfo {author} {\bibfnamefont {J.}~\bibnamefont {Salmilehto}}, \bibinfo {author} {\bibfnamefont {L.}~\bibnamefont {Jiang}},\ and\ \bibinfo {author} {\bibfnamefont {S.~M.}\ \bibnamefont {Girvin}},\ }\href {https://doi.org/10.1103/PhysRevX.6.031006} {\bibfield  {journal} {\bibinfo  {journal} {Phys. Rev. X}\ }\textbf {\bibinfo {volume} {6}},\ \bibinfo {pages} {031006} (\bibinfo {year} {2016})}\BibitemShut {NoStop}%
    \bibitem [{\citenamefont {Xu}\ \emph {et~al.}(2023{\natexlab{a}})\citenamefont {Xu}, \citenamefont {Zheng}, \citenamefont {Wang}, \citenamefont {Zoller}, \citenamefont {Clerk},\ and\ \citenamefont {Jiang}}]{XuNpjqi2023}%
      \BibitemOpen
      \bibfield  {author} {\bibinfo {author} {\bibfnamefont {Q.}~\bibnamefont {Xu}}, \bibinfo {author} {\bibfnamefont {G.}~\bibnamefont {Zheng}}, \bibinfo {author} {\bibfnamefont {Y.-X.}\ \bibnamefont {Wang}}, \bibinfo {author} {\bibfnamefont {P.}~\bibnamefont {Zoller}}, \bibinfo {author} {\bibfnamefont {A.~A.}\ \bibnamefont {Clerk}},\ and\ \bibinfo {author} {\bibfnamefont {L.}~\bibnamefont {Jiang}},\ }\href {https://doi.org/10.1038/s41534-023-00746-0} {\bibfield  {journal} {\bibinfo  {journal} {npj Quantum Information}\ }\textbf {\bibinfo {volume} {9}},\ \bibinfo {pages} {78} (\bibinfo {year} {2023}{\natexlab{a}})}\BibitemShut {NoStop}%
    \bibitem [{\citenamefont {Hashim}\ \emph {et~al.}(2021)\citenamefont {Hashim}, \citenamefont {Naik}, \citenamefont {Morvan}, \citenamefont {Ville}, \citenamefont {Mitchell}, \citenamefont {Kreikebaum}, \citenamefont {Davis}, \citenamefont {Smith}, \citenamefont {Iancu}, \citenamefont {O'Brien}, \citenamefont {Hincks}, \citenamefont {Wallman}, \citenamefont {Emerson},\ and\ \citenamefont {Siddiqi}}]{PhysRevX.11.041039}%
      \BibitemOpen
      \bibfield  {author} {\bibinfo {author} {\bibfnamefont {A.}~\bibnamefont {Hashim}}, \bibinfo {author} {\bibfnamefont {R.~K.}\ \bibnamefont {Naik}}, \bibinfo {author} {\bibfnamefont {A.}~\bibnamefont {Morvan}}, \bibinfo {author} {\bibfnamefont {J.-L.}\ \bibnamefont {Ville}}, \bibinfo {author} {\bibfnamefont {B.}~\bibnamefont {Mitchell}}, \bibinfo {author} {\bibfnamefont {J.~M.}\ \bibnamefont {Kreikebaum}}, \bibinfo {author} {\bibfnamefont {M.}~\bibnamefont {Davis}}, \bibinfo {author} {\bibfnamefont {E.}~\bibnamefont {Smith}}, \bibinfo {author} {\bibfnamefont {C.}~\bibnamefont {Iancu}}, \bibinfo {author} {\bibfnamefont {K.~P.}\ \bibnamefont {O'Brien}}, \bibinfo {author} {\bibfnamefont {I.}~\bibnamefont {Hincks}}, \bibinfo {author} {\bibfnamefont {J.~J.}\ \bibnamefont {Wallman}}, \bibinfo {author} {\bibfnamefont {J.}~\bibnamefont {Emerson}},\ and\ \bibinfo {author} {\bibfnamefont {I.}~\bibnamefont {Siddiqi}},\ }\href {https://doi.org/10.1103/PhysRevX.11.041039} {\bibfield  {journal} {\bibinfo  {journal} {Phys.
      Rev. X}\ }\textbf {\bibinfo {volume} {11}},\ \bibinfo {pages} {041039} (\bibinfo {year} {2021})}\BibitemShut {NoStop}%
    \bibitem [{\citenamefont {Wallman}\ and\ \citenamefont {Emerson}(2016)}]{PhysRevA.94.052325}%
      \BibitemOpen
      \bibfield  {author} {\bibinfo {author} {\bibfnamefont {J.~J.}\ \bibnamefont {Wallman}}\ and\ \bibinfo {author} {\bibfnamefont {J.}~\bibnamefont {Emerson}},\ }\href {https://doi.org/10.1103/PhysRevA.94.052325} {\bibfield  {journal} {\bibinfo  {journal} {Phys. Rev. A}\ }\textbf {\bibinfo {volume} {94}},\ \bibinfo {pages} {052325} (\bibinfo {year} {2016})}\BibitemShut {NoStop}%
    \bibitem [{\citenamefont {Albert}\ \emph {et~al.}(2018)\citenamefont {Albert}, \citenamefont {Noh}, \citenamefont {Duivenvoorden}, \citenamefont {Young}, \citenamefont {Brierley}, \citenamefont {Reinhold}, \citenamefont {Vuillot}, \citenamefont {Li}, \citenamefont {Shen}, \citenamefont {Girvin}, \citenamefont {Terhal},\ and\ \citenamefont {Jiang}}]{PhysRevA.97.032346}%
      \BibitemOpen
      \bibfield  {author} {\bibinfo {author} {\bibfnamefont {V.~V.}\ \bibnamefont {Albert}}, \bibinfo {author} {\bibfnamefont {K.}~\bibnamefont {Noh}}, \bibinfo {author} {\bibfnamefont {K.}~\bibnamefont {Duivenvoorden}}, \bibinfo {author} {\bibfnamefont {D.~J.}\ \bibnamefont {Young}}, \bibinfo {author} {\bibfnamefont {R.~T.}\ \bibnamefont {Brierley}}, \bibinfo {author} {\bibfnamefont {P.}~\bibnamefont {Reinhold}}, \bibinfo {author} {\bibfnamefont {C.}~\bibnamefont {Vuillot}}, \bibinfo {author} {\bibfnamefont {L.}~\bibnamefont {Li}}, \bibinfo {author} {\bibfnamefont {C.}~\bibnamefont {Shen}}, \bibinfo {author} {\bibfnamefont {S.~M.}\ \bibnamefont {Girvin}}, \bibinfo {author} {\bibfnamefont {B.~M.}\ \bibnamefont {Terhal}},\ and\ \bibinfo {author} {\bibfnamefont {L.}~\bibnamefont {Jiang}},\ }\href {https://doi.org/10.1103/PhysRevA.97.032346} {\bibfield  {journal} {\bibinfo  {journal} {Phys. Rev. A}\ }\textbf {\bibinfo {volume} {97}},\ \bibinfo {pages} {032346} (\bibinfo {year} {2018})}\BibitemShut {NoStop}%
    \bibitem [{\citenamefont {Reimpell}\ and\ \citenamefont {Werner}(2005)}]{PhysRevLett.94.080501}%
      \BibitemOpen
      \bibfield  {author} {\bibinfo {author} {\bibfnamefont {M.}~\bibnamefont {Reimpell}}\ and\ \bibinfo {author} {\bibfnamefont {R.~F.}\ \bibnamefont {Werner}},\ }\href {https://doi.org/10.1103/PhysRevLett.94.080501} {\bibfield  {journal} {\bibinfo  {journal} {Phys. Rev. Lett.}\ }\textbf {\bibinfo {volume} {94}},\ \bibinfo {pages} {080501} (\bibinfo {year} {2005})}\BibitemShut {NoStop}%
    \bibitem [{\citenamefont {Fletcher}\ \emph {et~al.}(2007)\citenamefont {Fletcher}, \citenamefont {Shor},\ and\ \citenamefont {Win}}]{PhysRevA.75.012338}%
      \BibitemOpen
      \bibfield  {author} {\bibinfo {author} {\bibfnamefont {A.~S.}\ \bibnamefont {Fletcher}}, \bibinfo {author} {\bibfnamefont {P.~W.}\ \bibnamefont {Shor}},\ and\ \bibinfo {author} {\bibfnamefont {M.~Z.}\ \bibnamefont {Win}},\ }\href {https://doi.org/10.1103/PhysRevA.75.012338} {\bibfield  {journal} {\bibinfo  {journal} {Phys. Rev. A}\ }\textbf {\bibinfo {volume} {75}},\ \bibinfo {pages} {012338} (\bibinfo {year} {2007})}\BibitemShut {NoStop}%
    \bibitem [{\citenamefont {Denys}\ and\ \citenamefont {Leverrier}(2023)}]{Denys2023tqutrittwomode}%
      \BibitemOpen
      \bibfield  {author} {\bibinfo {author} {\bibfnamefont {A.}~\bibnamefont {Denys}}\ and\ \bibinfo {author} {\bibfnamefont {A.}~\bibnamefont {Leverrier}},\ }\href {https://doi.org/10.22331/q-2023-06-05-1032} {\bibfield  {journal} {\bibinfo  {journal} {{Quantum}}\ }\textbf {\bibinfo {volume} {7}},\ \bibinfo {pages} {1032} (\bibinfo {year} {2023})}\BibitemShut {NoStop}%
    \bibitem [{\citenamefont {Schumacher}(1996)}]{PhysRevA.54.2614}%
      \BibitemOpen
      \bibfield  {author} {\bibinfo {author} {\bibfnamefont {B.}~\bibnamefont {Schumacher}},\ }\href {https://doi.org/10.1103/PhysRevA.54.2614} {\bibfield  {journal} {\bibinfo  {journal} {Phys. Rev. A}\ }\textbf {\bibinfo {volume} {54}},\ \bibinfo {pages} {2614} (\bibinfo {year} {1996})}\BibitemShut {NoStop}%
    \bibitem [{Note1()}]{Note1}%
      \BibitemOpen
      \bibinfo {note} {The channel fidelity is equivalent to the entanglement fidelity for maximally entangled states}\BibitemShut {NoStop}%
    \bibitem [{\citenamefont {Audenaert}\ and\ \citenamefont {De~Moor}(2002)}]{PhysRevA.65.030302}%
      \BibitemOpen
      \bibfield  {author} {\bibinfo {author} {\bibfnamefont {K.}~\bibnamefont {Audenaert}}\ and\ \bibinfo {author} {\bibfnamefont {B.}~\bibnamefont {De~Moor}},\ }\href {https://doi.org/10.1103/PhysRevA.65.030302} {\bibfield  {journal} {\bibinfo  {journal} {Phys. Rev. A}\ }\textbf {\bibinfo {volume} {65}},\ \bibinfo {pages} {030302} (\bibinfo {year} {2002})}\BibitemShut {NoStop}%
    \bibitem [{\citenamefont {Kosut}\ and\ \citenamefont {Lidar}(2009)}]{Kosut_2009}%
      \BibitemOpen
      \bibfield  {author} {\bibinfo {author} {\bibfnamefont {R.~L.}\ \bibnamefont {Kosut}}\ and\ \bibinfo {author} {\bibfnamefont {D.~A.}\ \bibnamefont {Lidar}},\ }\href {https://doi.org/10.1007/s11128-009-0120-2} {\bibfield  {journal} {\bibinfo  {journal} {Quantum Information Processing}\ }\textbf {\bibinfo {volume} {8}},\ \bibinfo {pages} {443–459} (\bibinfo {year} {2009})}\BibitemShut {NoStop}%
    \bibitem [{\citenamefont {Noh}\ \emph {et~al.}(2019)\citenamefont {Noh}, \citenamefont {Albert},\ and\ \citenamefont {Jiang}}]{8482307}%
      \BibitemOpen
      \bibfield  {author} {\bibinfo {author} {\bibfnamefont {K.}~\bibnamefont {Noh}}, \bibinfo {author} {\bibfnamefont {V.~V.}\ \bibnamefont {Albert}},\ and\ \bibinfo {author} {\bibfnamefont {L.}~\bibnamefont {Jiang}},\ }\href {https://doi.org/10.1109/TIT.2018.2873764} {\bibfield  {journal} {\bibinfo  {journal} {IEEE Transactions on Information Theory}\ }\textbf {\bibinfo {volume} {65}},\ \bibinfo {pages} {2563} (\bibinfo {year} {2019})}\BibitemShut {NoStop}%
    \bibitem [{\citenamefont {Schlegel}\ \emph {et~al.}(2022)\citenamefont {Schlegel}, \citenamefont {Minganti},\ and\ \citenamefont {Savona}}]{PhysRevA.106.022431}%
      \BibitemOpen
      \bibfield  {author} {\bibinfo {author} {\bibfnamefont {D.~S.}\ \bibnamefont {Schlegel}}, \bibinfo {author} {\bibfnamefont {F.}~\bibnamefont {Minganti}},\ and\ \bibinfo {author} {\bibfnamefont {V.}~\bibnamefont {Savona}},\ }\href {https://doi.org/10.1103/PhysRevA.106.022431} {\bibfield  {journal} {\bibinfo  {journal} {Phys. Rev. A}\ }\textbf {\bibinfo {volume} {106}},\ \bibinfo {pages} {022431} (\bibinfo {year} {2022})}\BibitemShut {NoStop}%
    \bibitem [{\citenamefont {Fletcher}(2007)}]{fletcher2007channeladapted}%
      \BibitemOpen
      \bibfield  {author} {\bibinfo {author} {\bibfnamefont {A.~S.}\ \bibnamefont {Fletcher}},\ }\href@noop {} {\  (\bibinfo {year} {2007})},\ \Eprint {https://arxiv.org/abs/0706.3400} {arXiv:0706.3400 [quant-ph]} \BibitemShut {NoStop}%
    \bibitem [{\citenamefont {Reimpell}\ \emph {et~al.}(2006)\citenamefont {Reimpell}, \citenamefont {Werner},\ and\ \citenamefont {Audenaert}}]{reimpell2006comment}%
      \BibitemOpen
      \bibfield  {author} {\bibinfo {author} {\bibfnamefont {M.}~\bibnamefont {Reimpell}}, \bibinfo {author} {\bibfnamefont {R.~F.}\ \bibnamefont {Werner}},\ and\ \bibinfo {author} {\bibfnamefont {K.}~\bibnamefont {Audenaert}},\ }\href@noop {} {\  (\bibinfo {year} {2006})},\ \Eprint {https://arxiv.org/abs/quant-ph/0606059} {arXiv:quant-ph/0606059 [quant-ph]} \BibitemShut {NoStop}%
    \bibitem [{\citenamefont {Yamamoto}\ \emph {et~al.}(2005)\citenamefont {Yamamoto}, \citenamefont {Hara},\ and\ \citenamefont {Tsumura}}]{Yamamoto_2005}%
      \BibitemOpen
      \bibfield  {author} {\bibinfo {author} {\bibfnamefont {N.}~\bibnamefont {Yamamoto}}, \bibinfo {author} {\bibfnamefont {S.}~\bibnamefont {Hara}},\ and\ \bibinfo {author} {\bibfnamefont {K.}~\bibnamefont {Tsumura}},\ }\bibfield  {journal} {\bibinfo  {journal} {Physical Review A}\ }\textbf {\bibinfo {volume} {71}},\ \href {https://doi.org/10.1103/physreva.71.022322} {10.1103/physreva.71.022322} (\bibinfo {year} {2005})\BibitemShut {NoStop}%
    \bibitem [{\citenamefont {B\'eny}\ and\ \citenamefont {Oreshkov}(2010)}]{PhysRevLett.104.120501}%
      \BibitemOpen
      \bibfield  {author} {\bibinfo {author} {\bibfnamefont {C.}~\bibnamefont {B\'eny}}\ and\ \bibinfo {author} {\bibfnamefont {O.}~\bibnamefont {Oreshkov}},\ }\href {https://doi.org/10.1103/PhysRevLett.104.120501} {\bibfield  {journal} {\bibinfo  {journal} {Phys. Rev. Lett.}\ }\textbf {\bibinfo {volume} {104}},\ \bibinfo {pages} {120501} (\bibinfo {year} {2010})}\BibitemShut {NoStop}%
    \bibitem [{\citenamefont {Xu}\ \emph {et~al.}(2023{\natexlab{b}})\citenamefont {Xu}, \citenamefont {Zheng}, \citenamefont {Wang}, \citenamefont {Zoller}, \citenamefont {Clerk},\ and\ \citenamefont {Jiang}}]{XuZheng2023}%
      \BibitemOpen
      \bibfield  {author} {\bibinfo {author} {\bibfnamefont {Q.}~\bibnamefont {Xu}}, \bibinfo {author} {\bibfnamefont {G.}~\bibnamefont {Zheng}}, \bibinfo {author} {\bibfnamefont {Y.-X.}\ \bibnamefont {Wang}}, \bibinfo {author} {\bibfnamefont {P.}~\bibnamefont {Zoller}}, \bibinfo {author} {\bibfnamefont {A.~A.}\ \bibnamefont {Clerk}},\ and\ \bibinfo {author} {\bibfnamefont {L.}~\bibnamefont {Jiang}},\ }\href {https://doi.org/10.1038/s41534-023-00746-0} {\bibfield  {journal} {\bibinfo  {journal} {npj Quantum Information}\ }\textbf {\bibinfo {volume} {9}},\ \bibinfo {pages} {78} (\bibinfo {year} {2023}{\natexlab{b}})}\BibitemShut {NoStop}%
    \bibitem [{\citenamefont {Kosut}\ \emph {et~al.}(2008)\citenamefont {Kosut}, \citenamefont {Shabani},\ and\ \citenamefont {Lidar}}]{PhysRevLett.100.020502}%
      \BibitemOpen
      \bibfield  {author} {\bibinfo {author} {\bibfnamefont {R.~L.}\ \bibnamefont {Kosut}}, \bibinfo {author} {\bibfnamefont {A.}~\bibnamefont {Shabani}},\ and\ \bibinfo {author} {\bibfnamefont {D.~A.}\ \bibnamefont {Lidar}},\ }\href {https://doi.org/10.1103/PhysRevLett.100.020502} {\bibfield  {journal} {\bibinfo  {journal} {Phys. Rev. Lett.}\ }\textbf {\bibinfo {volume} {100}},\ \bibinfo {pages} {020502} (\bibinfo {year} {2008})}\BibitemShut {NoStop}%
    \bibitem [{\citenamefont {Jayashankar}\ \emph {et~al.}(2020)\citenamefont {Jayashankar}, \citenamefont {Babu}, \citenamefont {Ng},\ and\ \citenamefont {Mandayam}}]{PhysRevA.101.042307}%
      \BibitemOpen
      \bibfield  {author} {\bibinfo {author} {\bibfnamefont {A.}~\bibnamefont {Jayashankar}}, \bibinfo {author} {\bibfnamefont {A.~M.}\ \bibnamefont {Babu}}, \bibinfo {author} {\bibfnamefont {H.~K.}\ \bibnamefont {Ng}},\ and\ \bibinfo {author} {\bibfnamefont {P.}~\bibnamefont {Mandayam}},\ }\href {https://doi.org/10.1103/PhysRevA.101.042307} {\bibfield  {journal} {\bibinfo  {journal} {Phys. Rev. A}\ }\textbf {\bibinfo {volume} {101}},\ \bibinfo {pages} {042307} (\bibinfo {year} {2020})}\BibitemShut {NoStop}%
    \bibitem [{\citenamefont {Leviant}\ \emph {et~al.}(2022)\citenamefont {Leviant}, \citenamefont {Xu}, \citenamefont {Jiang},\ and\ \citenamefont {Rosenblum}}]{Leviant2022quantumcapacity}%
      \BibitemOpen
      \bibfield  {author} {\bibinfo {author} {\bibfnamefont {P.}~\bibnamefont {Leviant}}, \bibinfo {author} {\bibfnamefont {Q.}~\bibnamefont {Xu}}, \bibinfo {author} {\bibfnamefont {L.}~\bibnamefont {Jiang}},\ and\ \bibinfo {author} {\bibfnamefont {S.}~\bibnamefont {Rosenblum}},\ }\href {https://doi.org/10.22331/q-2022-09-29-821} {\bibfield  {journal} {\bibinfo  {journal} {{Quantum}}\ }\textbf {\bibinfo {volume} {6}},\ \bibinfo {pages} {821} (\bibinfo {year} {2022})}\BibitemShut {NoStop}%
    \bibitem [{\citenamefont {Barnum}\ and\ \citenamefont {Knill}(2002)}]{10.1063/1.1459754}%
      \BibitemOpen
      \bibfield  {author} {\bibinfo {author} {\bibfnamefont {H.}~\bibnamefont {Barnum}}\ and\ \bibinfo {author} {\bibfnamefont {E.}~\bibnamefont {Knill}},\ }\href {https://doi.org/10.1063/1.1459754} {\bibfield  {journal} {\bibinfo  {journal} {Journal of Mathematical Physics}\ }\textbf {\bibinfo {volume} {43}},\ \bibinfo {pages} {2097} (\bibinfo {year} {2002})}\BibitemShut {NoStop}%
    \bibitem [{\citenamefont {Ng}\ and\ \citenamefont {Mandayam}(2010)}]{PhysRevA.81.062342}%
      \BibitemOpen
      \bibfield  {author} {\bibinfo {author} {\bibfnamefont {H.~K.}\ \bibnamefont {Ng}}\ and\ \bibinfo {author} {\bibfnamefont {P.}~\bibnamefont {Mandayam}},\ }\href {https://doi.org/10.1103/PhysRevA.81.062342} {\bibfield  {journal} {\bibinfo  {journal} {Phys. Rev. A}\ }\textbf {\bibinfo {volume} {81}},\ \bibinfo {pages} {062342} (\bibinfo {year} {2010})}\BibitemShut {NoStop}%
    \bibitem [{\citenamefont {Petz}(1988)}]{Petz:1988usv}%
      \BibitemOpen
      \bibfield  {author} {\bibinfo {author} {\bibfnamefont {D.}~\bibnamefont {Petz}},\ }\href {https://doi.org/10.1093/qmath/39.1.97} {\bibfield  {journal} {\bibinfo  {journal} {Quart. J. Math. Oxford Ser.}\ }\textbf {\bibinfo {volume} {39}},\ \bibinfo {pages} {97} (\bibinfo {year} {1988})}\BibitemShut {NoStop}%
    \bibitem [{\citenamefont {Wilde}(2017)}]{wilde2017quantum}%
      \BibitemOpen
      \bibfield  {author} {\bibinfo {author} {\bibfnamefont {M.}~\bibnamefont {Wilde}},\ }\href {https://books.google.com/books?id=gYcHDgAAQBAJ} {\emph {\bibinfo {title} {Quantum Information Theory}}}\ (\bibinfo  {publisher} {Cambridge University Press},\ \bibinfo {year} {2017})\BibitemShut {NoStop}%
    \bibitem [{\citenamefont {Tyson}(2010)}]{10.1063/1.3463451}%
      \BibitemOpen
      \bibfield  {author} {\bibinfo {author} {\bibfnamefont {J.}~\bibnamefont {Tyson}},\ }\href {https://doi.org/10.1063/1.3463451} {\bibfield  {journal} {\bibinfo  {journal} {Journal of Mathematical Physics}\ }\textbf {\bibinfo {volume} {51}},\ \bibinfo {pages} {092204} (\bibinfo {year} {2010})}\BibitemShut {NoStop}%
    \bibitem [{\citenamefont {B\'eny}(2011)}]{PhysRevLett.107.080501}%
      \BibitemOpen
      \bibfield  {author} {\bibinfo {author} {\bibfnamefont {C.}~\bibnamefont {B\'eny}},\ }\href {https://doi.org/10.1103/PhysRevLett.107.080501} {\bibfield  {journal} {\bibinfo  {journal} {Phys. Rev. Lett.}\ }\textbf {\bibinfo {volume} {107}},\ \bibinfo {pages} {080501} (\bibinfo {year} {2011})}\BibitemShut {NoStop}%
    \bibitem [{\citenamefont {Bény}\ and\ \citenamefont {Oreshkov}(2011)}]{B_ny_2011}%
      \BibitemOpen
      \bibfield  {author} {\bibinfo {author} {\bibfnamefont {C.}~\bibnamefont {Bény}}\ and\ \bibinfo {author} {\bibfnamefont {O.}~\bibnamefont {Oreshkov}},\ }\bibfield  {journal} {\bibinfo  {journal} {Physical Review A}\ }\textbf {\bibinfo {volume} {84}},\ \href {https://doi.org/10.1103/physreva.84.022333} {10.1103/physreva.84.022333} (\bibinfo {year} {2011})\BibitemShut {NoStop}%
    \bibitem [{\citenamefont {Mandayam}\ and\ \citenamefont {Ng}(2012)}]{PhysRevA.86.012335}%
      \BibitemOpen
      \bibfield  {author} {\bibinfo {author} {\bibfnamefont {P.}~\bibnamefont {Mandayam}}\ and\ \bibinfo {author} {\bibfnamefont {H.~K.}\ \bibnamefont {Ng}},\ }\href {https://doi.org/10.1103/PhysRevA.86.012335} {\bibfield  {journal} {\bibinfo  {journal} {Phys. Rev. A}\ }\textbf {\bibinfo {volume} {86}},\ \bibinfo {pages} {012335} (\bibinfo {year} {2012})}\BibitemShut {NoStop}%
    \bibitem [{\citenamefont {Junge}\ \emph {et~al.}(2018)\citenamefont {Junge}, \citenamefont {Renner}, \citenamefont {Sutter}, \citenamefont {Wilde},\ and\ \citenamefont {Winter}}]{Junge_2018}%
      \BibitemOpen
      \bibfield  {author} {\bibinfo {author} {\bibfnamefont {M.}~\bibnamefont {Junge}}, \bibinfo {author} {\bibfnamefont {R.}~\bibnamefont {Renner}}, \bibinfo {author} {\bibfnamefont {D.}~\bibnamefont {Sutter}}, \bibinfo {author} {\bibfnamefont {M.~M.}\ \bibnamefont {Wilde}},\ and\ \bibinfo {author} {\bibfnamefont {A.}~\bibnamefont {Winter}},\ }\href {https://doi.org/10.1007/s00023-018-0716-0} {\bibfield  {journal} {\bibinfo  {journal} {Annales Henri Poincaré}\ }\textbf {\bibinfo {volume} {19}},\ \bibinfo {pages} {2955–2978} (\bibinfo {year} {2018})}\BibitemShut {NoStop}%
    \bibitem [{\citenamefont {Brand\~ao}\ \emph {et~al.}(2019)\citenamefont {Brand\~ao}, \citenamefont {Crosson}, \citenamefont {\ifmmode \mbox{\c{S}}\else \c{S}\fi{}ahino\ifmmode~\breve{g}\else \u{g}\fi{}lu},\ and\ \citenamefont {Bowen}}]{PhysRevLett.123.110502}%
      \BibitemOpen
      \bibfield  {author} {\bibinfo {author} {\bibfnamefont {F.~G. S.~L.}\ \bibnamefont {Brand\~ao}}, \bibinfo {author} {\bibfnamefont {E.}~\bibnamefont {Crosson}}, \bibinfo {author} {\bibfnamefont {M.~B.}\ \bibnamefont {\ifmmode \mbox{\c{S}}\else \c{S}\fi{}ahino\ifmmode~\breve{g}\else \u{g}\fi{}lu}},\ and\ \bibinfo {author} {\bibfnamefont {J.}~\bibnamefont {Bowen}},\ }\href {https://doi.org/10.1103/PhysRevLett.123.110502} {\bibfield  {journal} {\bibinfo  {journal} {Phys. Rev. Lett.}\ }\textbf {\bibinfo {volume} {123}},\ \bibinfo {pages} {110502} (\bibinfo {year} {2019})}\BibitemShut {NoStop}%
    \bibitem [{\citenamefont {Faist}\ \emph {et~al.}(2020)\citenamefont {Faist}, \citenamefont {Nezami}, \citenamefont {Albert}, \citenamefont {Salton}, \citenamefont {Pastawski}, \citenamefont {Hayden},\ and\ \citenamefont {Preskill}}]{PhysRevX.10.041018}%
      \BibitemOpen
      \bibfield  {author} {\bibinfo {author} {\bibfnamefont {P.}~\bibnamefont {Faist}}, \bibinfo {author} {\bibfnamefont {S.}~\bibnamefont {Nezami}}, \bibinfo {author} {\bibfnamefont {V.~V.}\ \bibnamefont {Albert}}, \bibinfo {author} {\bibfnamefont {G.}~\bibnamefont {Salton}}, \bibinfo {author} {\bibfnamefont {F.}~\bibnamefont {Pastawski}}, \bibinfo {author} {\bibfnamefont {P.}~\bibnamefont {Hayden}},\ and\ \bibinfo {author} {\bibfnamefont {J.}~\bibnamefont {Preskill}},\ }\href {https://doi.org/10.1103/PhysRevX.10.041018} {\bibfield  {journal} {\bibinfo  {journal} {Phys. Rev. X}\ }\textbf {\bibinfo {volume} {10}},\ \bibinfo {pages} {041018} (\bibinfo {year} {2020})}\BibitemShut {NoStop}%
    \bibitem [{\citenamefont {Zhou}\ \emph {et~al.}(2021)\citenamefont {Zhou}, \citenamefont {Liu},\ and\ \citenamefont {Jiang}}]{Zhou2021newperspectives}%
      \BibitemOpen
      \bibfield  {author} {\bibinfo {author} {\bibfnamefont {S.}~\bibnamefont {Zhou}}, \bibinfo {author} {\bibfnamefont {Z.-W.}\ \bibnamefont {Liu}},\ and\ \bibinfo {author} {\bibfnamefont {L.}~\bibnamefont {Jiang}},\ }\href {https://doi.org/10.22331/q-2021-08-09-521} {\bibfield  {journal} {\bibinfo  {journal} {{Quantum}}\ }\textbf {\bibinfo {volume} {5}},\ \bibinfo {pages} {521} (\bibinfo {year} {2021})}\BibitemShut {NoStop}%
    \bibitem [{\citenamefont {Gily\'en}\ \emph {et~al.}(2022)\citenamefont {Gily\'en}, \citenamefont {Lloyd}, \citenamefont {Marvian}, \citenamefont {Quek},\ and\ \citenamefont {Wilde}}]{PhysRevLett.128.220502}%
      \BibitemOpen
      \bibfield  {author} {\bibinfo {author} {\bibfnamefont {A.}~\bibnamefont {Gily\'en}}, \bibinfo {author} {\bibfnamefont {S.}~\bibnamefont {Lloyd}}, \bibinfo {author} {\bibfnamefont {I.}~\bibnamefont {Marvian}}, \bibinfo {author} {\bibfnamefont {Y.}~\bibnamefont {Quek}},\ and\ \bibinfo {author} {\bibfnamefont {M.~M.}\ \bibnamefont {Wilde}},\ }\href {https://doi.org/10.1103/PhysRevLett.128.220502} {\bibfield  {journal} {\bibinfo  {journal} {Phys. Rev. Lett.}\ }\textbf {\bibinfo {volume} {128}},\ \bibinfo {pages} {220502} (\bibinfo {year} {2022})}\BibitemShut {NoStop}%
    \bibitem [{\citenamefont {Horodecki}\ \emph {et~al.}(1999)\citenamefont {Horodecki}, \citenamefont {Horodecki},\ and\ \citenamefont {Horodecki}}]{PhysRevA.60.1888}%
      \BibitemOpen
      \bibfield  {author} {\bibinfo {author} {\bibfnamefont {M.}~\bibnamefont {Horodecki}}, \bibinfo {author} {\bibfnamefont {P.}~\bibnamefont {Horodecki}},\ and\ \bibinfo {author} {\bibfnamefont {R.}~\bibnamefont {Horodecki}},\ }\href {https://doi.org/10.1103/PhysRevA.60.1888} {\bibfield  {journal} {\bibinfo  {journal} {Phys. Rev. A}\ }\textbf {\bibinfo {volume} {60}},\ \bibinfo {pages} {1888} (\bibinfo {year} {1999})}\BibitemShut {NoStop}%
    \bibitem [{\citenamefont {Nielsen}(2002)}]{Nielsen_2002}%
      \BibitemOpen
      \bibfield  {author} {\bibinfo {author} {\bibfnamefont {M.~A.}\ \bibnamefont {Nielsen}},\ }\href {https://doi.org/10.1016/s0375-9601(02)01272-0} {\bibfield  {journal} {\bibinfo  {journal} {Physics Letters A}\ }\textbf {\bibinfo {volume} {303}},\ \bibinfo {pages} {249–252} (\bibinfo {year} {2002})}\BibitemShut {NoStop}%
    \bibitem [{SM()}]{SM}%
      \BibitemOpen
      \href@noop {} {}\bibinfo {note} {See Supplemental Material for the formulation of optimization-based QEC, derivation of analytical results, and the definitions of noise channels and codes, which includes Ref.~\cite{Caves1999, nielsen_chuang_2010, 1975RaEl...20.1177B, doi:10.1137/1123048, Hausladen1994AG, PhysRevA.77.034101, 10.1116/5.0060893, PhysRevX.9.031029, Buscemi_2021, biswas2023noiseadapted, PhysRevB.95.134501, PhysRevA.95.052316, Rosenblum_2018, PhysRevLett.125.110503}}\BibitemShut {NoStop}%
    \bibitem [{\citenamefont {Berta}\ and\ \citenamefont {Tomamichel}(2016)}]{7404264}%
      \BibitemOpen
      \bibfield  {author} {\bibinfo {author} {\bibfnamefont {M.}~\bibnamefont {Berta}}\ and\ \bibinfo {author} {\bibfnamefont {M.}~\bibnamefont {Tomamichel}},\ }\href {https://doi.org/10.1109/TIT.2016.2527683} {\bibfield  {journal} {\bibinfo  {journal} {IEEE Transactions on Information Theory}\ }\textbf {\bibinfo {volume} {62}},\ \bibinfo {pages} {1758} (\bibinfo {year} {2016})}\BibitemShut {NoStop}%
    \bibitem [{\citenamefont {Karimipour}\ \emph {et~al.}(2020)\citenamefont {Karimipour}, \citenamefont {Benatti},\ and\ \citenamefont {Floreanini}}]{PhysRevA.101.032109}%
      \BibitemOpen
      \bibfield  {author} {\bibinfo {author} {\bibfnamefont {V.}~\bibnamefont {Karimipour}}, \bibinfo {author} {\bibfnamefont {F.}~\bibnamefont {Benatti}},\ and\ \bibinfo {author} {\bibfnamefont {R.}~\bibnamefont {Floreanini}},\ }\href {https://doi.org/10.1103/PhysRevA.101.032109} {\bibfield  {journal} {\bibinfo  {journal} {Phys. Rev. A}\ }\textbf {\bibinfo {volume} {101}},\ \bibinfo {pages} {032109} (\bibinfo {year} {2020})}\BibitemShut {NoStop}%
    \bibitem [{\citenamefont {Kribs}\ \emph {et~al.}(2005)\citenamefont {Kribs}, \citenamefont {Laflamme},\ and\ \citenamefont {Poulin}}]{PhysRevLett.94.180501}%
      \BibitemOpen
      \bibfield  {author} {\bibinfo {author} {\bibfnamefont {D.}~\bibnamefont {Kribs}}, \bibinfo {author} {\bibfnamefont {R.}~\bibnamefont {Laflamme}},\ and\ \bibinfo {author} {\bibfnamefont {D.}~\bibnamefont {Poulin}},\ }\href {https://doi.org/10.1103/PhysRevLett.94.180501} {\bibfield  {journal} {\bibinfo  {journal} {Phys. Rev. Lett.}\ }\textbf {\bibinfo {volume} {94}},\ \bibinfo {pages} {180501} (\bibinfo {year} {2005})}\BibitemShut {NoStop}%
    \bibitem [{\citenamefont {Koashi}\ and\ \citenamefont {Ueda}(1999)}]{PhysRevLett.82.2598}%
      \BibitemOpen
      \bibfield  {author} {\bibinfo {author} {\bibfnamefont {M.}~\bibnamefont {Koashi}}\ and\ \bibinfo {author} {\bibfnamefont {M.}~\bibnamefont {Ueda}},\ }\href {https://doi.org/10.1103/PhysRevLett.82.2598} {\bibfield  {journal} {\bibinfo  {journal} {Phys. Rev. Lett.}\ }\textbf {\bibinfo {volume} {82}},\ \bibinfo {pages} {2598} (\bibinfo {year} {1999})}\BibitemShut {NoStop}%
    \bibitem [{\citenamefont {Mohan}\ \emph {et~al.}(2023)\citenamefont {Mohan}, \citenamefont {Sikora},\ and\ \citenamefont {Upadhyay}}]{mohan2023generalized}%
      \BibitemOpen
      \bibfield  {author} {\bibinfo {author} {\bibfnamefont {A.}~\bibnamefont {Mohan}}, \bibinfo {author} {\bibfnamefont {J.}~\bibnamefont {Sikora}},\ and\ \bibinfo {author} {\bibfnamefont {S.}~\bibnamefont {Upadhyay}},\ }\href@noop {} {\  (\bibinfo {year} {2023})},\ \Eprint {https://arxiv.org/abs/2312.04023} {arXiv:2312.04023 [quant-ph]} \BibitemShut {NoStop}%
    \bibitem [{\citenamefont {Johnston}\ \emph {et~al.}(2023)\citenamefont {Johnston}, \citenamefont {Russo},\ and\ \citenamefont {Sikora}}]{johnston2023tight}%
      \BibitemOpen
      \bibfield  {author} {\bibinfo {author} {\bibfnamefont {N.}~\bibnamefont {Johnston}}, \bibinfo {author} {\bibfnamefont {V.}~\bibnamefont {Russo}},\ and\ \bibinfo {author} {\bibfnamefont {J.}~\bibnamefont {Sikora}},\ }\href@noop {} {\  (\bibinfo {year} {2023})},\ \Eprint {https://arxiv.org/abs/2311.17047} {arXiv:2311.17047 [quant-ph]} \BibitemShut {NoStop}%
    \bibitem [{\citenamefont {Jiang}\ \emph {et~al.}(2020)\citenamefont {Jiang}, \citenamefont {Kathuria}, \citenamefont {Lee}, \citenamefont {Padmanabhan},\ and\ \citenamefont {Song}}]{9317892}%
      \BibitemOpen
      \bibfield  {author} {\bibinfo {author} {\bibfnamefont {H.}~\bibnamefont {Jiang}}, \bibinfo {author} {\bibfnamefont {T.}~\bibnamefont {Kathuria}}, \bibinfo {author} {\bibfnamefont {Y.~T.}\ \bibnamefont {Lee}}, \bibinfo {author} {\bibfnamefont {S.}~\bibnamefont {Padmanabhan}},\ and\ \bibinfo {author} {\bibfnamefont {Z.}~\bibnamefont {Song}},\ }in\ \href {https://doi.org/10.1109/FOCS46700.2020.00089} {\emph {\bibinfo {booktitle} {2020 IEEE 61st Annual Symposium on Foundations of Computer Science (FOCS)}}}\ (\bibinfo {year} {2020})\ pp.\ \bibinfo {pages} {910--918}\BibitemShut {NoStop}%
    \bibitem [{Note2()}]{Note2}%
      \BibitemOpen
      \bibinfo {note} {$\protect \tilde {\protect \mathcal {O}}(g(n))$ is shortened for $\protect \mathcal {O}(g(n)\log ^k g(n))$}\BibitemShut {NoStop}%
    \bibitem [{\citenamefont {Daletskii}\ and\ \citenamefont {Krein}(1965)}]{DaletskiiKrein1965}%
      \BibitemOpen
      \bibfield  {author} {\bibinfo {author} {\bibfnamefont {J.~L.}\ \bibnamefont {Daletskii}}\ and\ \bibinfo {author} {\bibfnamefont {S.}~\bibnamefont {Krein}},\ }\href@noop {} {\bibfield  {journal} {\bibinfo  {journal} {AMS Translations}\ }\textbf {\bibinfo {volume} {2}},\ \bibinfo {pages} {(47)1} (\bibinfo {year} {1965})}\BibitemShut {NoStop}%
    \bibitem [{\citenamefont {Carlsson}(2018)}]{carlsson2018perturbation}%
      \BibitemOpen
      \bibfield  {author} {\bibinfo {author} {\bibfnamefont {M.}~\bibnamefont {Carlsson}},\ }\href@noop {} {\  (\bibinfo {year} {2018})},\ \Eprint {https://arxiv.org/abs/1810.01464} {arXiv:1810.01464 [math.FA]} \BibitemShut {NoStop}%
    \bibitem [{\citenamefont {Olle}\ \emph {et~al.}(2023)\citenamefont {Olle}, \citenamefont {Zen}, \citenamefont {Puviani},\ and\ \citenamefont {Marquardt}}]{olle2023simultaneous}%
      \BibitemOpen
      \bibfield  {author} {\bibinfo {author} {\bibfnamefont {J.}~\bibnamefont {Olle}}, \bibinfo {author} {\bibfnamefont {R.}~\bibnamefont {Zen}}, \bibinfo {author} {\bibfnamefont {M.}~\bibnamefont {Puviani}},\ and\ \bibinfo {author} {\bibfnamefont {F.}~\bibnamefont {Marquardt}},\ }\href@noop {} {\  (\bibinfo {year} {2023})},\ \Eprint {https://arxiv.org/abs/2311.04750} {arXiv:2311.04750 [quant-ph]} \BibitemShut {NoStop}%
    \bibitem [{\citenamefont {Zeng}\ \emph {et~al.}(2023)\citenamefont {Zeng}, \citenamefont {Zhou}, \citenamefont {Rinaldi}, \citenamefont {Gneiting},\ and\ \citenamefont {Nori}}]{PhysRevLett.131.050601}%
      \BibitemOpen
      \bibfield  {author} {\bibinfo {author} {\bibfnamefont {Y.}~\bibnamefont {Zeng}}, \bibinfo {author} {\bibfnamefont {Z.-Y.}\ \bibnamefont {Zhou}}, \bibinfo {author} {\bibfnamefont {E.}~\bibnamefont {Rinaldi}}, \bibinfo {author} {\bibfnamefont {C.}~\bibnamefont {Gneiting}},\ and\ \bibinfo {author} {\bibfnamefont {F.}~\bibnamefont {Nori}},\ }\href {https://doi.org/10.1103/PhysRevLett.131.050601} {\bibfield  {journal} {\bibinfo  {journal} {Phys. Rev. Lett.}\ }\textbf {\bibinfo {volume} {131}},\ \bibinfo {pages} {050601} (\bibinfo {year} {2023})}\BibitemShut {NoStop}%
    \bibitem [{\citenamefont {Korolev}\ \emph {et~al.}(2023)\citenamefont {Korolev}, \citenamefont {Bashmakova},\ and\ \citenamefont {Golubeva}}]{korolev2023error}%
      \BibitemOpen
      \bibfield  {author} {\bibinfo {author} {\bibfnamefont {S.~B.}\ \bibnamefont {Korolev}}, \bibinfo {author} {\bibfnamefont {E.~N.}\ \bibnamefont {Bashmakova}},\ and\ \bibinfo {author} {\bibfnamefont {T.~Y.}\ \bibnamefont {Golubeva}},\ }\href@noop {} {\  (\bibinfo {year} {2023})},\ \Eprint {https://arxiv.org/abs/2312.16000} {arXiv:2312.16000 [quant-ph]} \BibitemShut {NoStop}%
    \bibitem [{\citenamefont {Cao}\ \emph {et~al.}(2022)\citenamefont {Cao}, \citenamefont {Zhang}, \citenamefont {Wu}, \citenamefont {Grassl},\ and\ \citenamefont {Zeng}}]{Cao_2022}%
      \BibitemOpen
      \bibfield  {author} {\bibinfo {author} {\bibfnamefont {C.}~\bibnamefont {Cao}}, \bibinfo {author} {\bibfnamefont {C.}~\bibnamefont {Zhang}}, \bibinfo {author} {\bibfnamefont {Z.}~\bibnamefont {Wu}}, \bibinfo {author} {\bibfnamefont {M.}~\bibnamefont {Grassl}},\ and\ \bibinfo {author} {\bibfnamefont {B.}~\bibnamefont {Zeng}},\ }\href {https://doi.org/10.22331/q-2022-10-06-828} {\bibfield  {journal} {\bibinfo  {journal} {Quantum}\ }\textbf {\bibinfo {volume} {6}},\ \bibinfo {pages} {828} (\bibinfo {year} {2022})}\BibitemShut {NoStop}%
    \bibitem [{\citenamefont {Harrington}\ and\ \citenamefont {Preskill}(2001)}]{PhysRevA.64.062301}%
      \BibitemOpen
      \bibfield  {author} {\bibinfo {author} {\bibfnamefont {J.}~\bibnamefont {Harrington}}\ and\ \bibinfo {author} {\bibfnamefont {J.}~\bibnamefont {Preskill}},\ }\href {https://doi.org/10.1103/PhysRevA.64.062301} {\bibfield  {journal} {\bibinfo  {journal} {Phys. Rev. A}\ }\textbf {\bibinfo {volume} {64}},\ \bibinfo {pages} {062301} (\bibinfo {year} {2001})}\BibitemShut {NoStop}%
    \bibitem [{\citenamefont {Steane}(1996)}]{1996Steane}%
      \BibitemOpen
      \bibfield  {author} {\bibinfo {author} {\bibfnamefont {A.}~\bibnamefont {Steane}},\ }\href {https://doi.org/10.1098/rspa.1996.0136} {\bibfield  {journal} {\bibinfo  {journal} {Proceedings of the Royal Society of London. Series A: Mathematical, Physical and Engineering Sciences}\ }\textbf {\bibinfo {volume} {452}},\ \bibinfo {pages} {2551–2577} (\bibinfo {year} {1996})}\BibitemShut {NoStop}%
    \bibitem [{\citenamefont {Rains}\ \emph {et~al.}(1997)\citenamefont {Rains}, \citenamefont {Hardin}, \citenamefont {Shor},\ and\ \citenamefont {Sloane}}]{PhysRevLett.79.953}%
      \BibitemOpen
      \bibfield  {author} {\bibinfo {author} {\bibfnamefont {E.~M.}\ \bibnamefont {Rains}}, \bibinfo {author} {\bibfnamefont {R.~H.}\ \bibnamefont {Hardin}}, \bibinfo {author} {\bibfnamefont {P.~W.}\ \bibnamefont {Shor}},\ and\ \bibinfo {author} {\bibfnamefont {N.~J.~A.}\ \bibnamefont {Sloane}},\ }\href {https://doi.org/10.1103/PhysRevLett.79.953} {\bibfield  {journal} {\bibinfo  {journal} {Phys. Rev. Lett.}\ }\textbf {\bibinfo {volume} {79}},\ \bibinfo {pages} {953} (\bibinfo {year} {1997})}\BibitemShut {NoStop}%
    \bibitem [{\citenamefont {Srednicki}(1994)}]{PhysRevE.50.888}%
      \BibitemOpen
      \bibfield  {author} {\bibinfo {author} {\bibfnamefont {M.}~\bibnamefont {Srednicki}},\ }\href {https://doi.org/10.1103/PhysRevE.50.888} {\bibfield  {journal} {\bibinfo  {journal} {Phys. Rev. E}\ }\textbf {\bibinfo {volume} {50}},\ \bibinfo {pages} {888} (\bibinfo {year} {1994})}\BibitemShut {NoStop}%
    \bibitem [{\citenamefont {Zheng}\ \emph {et~al.}()\citenamefont {Zheng}, \citenamefont {He}, \citenamefont {Lee}, \citenamefont {Noh},\ and\ \citenamefont {Jiang}}]{OtherPaper}%
      \BibitemOpen
      \bibfield  {author} {\bibinfo {author} {\bibfnamefont {G.}~\bibnamefont {Zheng}}, \bibinfo {author} {\bibfnamefont {W.}~\bibnamefont {He}}, \bibinfo {author} {\bibfnamefont {G.}~\bibnamefont {Lee}}, \bibinfo {author} {\bibfnamefont {K.}~\bibnamefont {Noh}},\ and\ \bibinfo {author} {\bibfnamefont {L.}~\bibnamefont {Jiang}},\ }\href@noop {} {\bibinfo  {journal} {In preparation}\ }\BibitemShut {NoStop}%
    \bibitem [{\citenamefont {Caves}(1999)}]{Caves1999}%
      \BibitemOpen
    \bibfield  {journal} {  }\bibfield  {author} {\bibinfo {author} {\bibfnamefont {C.~M.}\ \bibnamefont {Caves}},\ }\href {https://doi.org/10.1023/A:1007720606911} {\bibfield  {journal} {\bibinfo  {journal} {Journal of Superconductivity}\ }\textbf {\bibinfo {volume} {12}},\ \bibinfo {pages} {707} (\bibinfo {year} {1999})}\BibitemShut {NoStop}%
    \bibitem [{\citenamefont {Nielsen}\ and\ \citenamefont {Chuang}(2010)}]{nielsen_chuang_2010}%
      \BibitemOpen
      \bibfield  {author} {\bibinfo {author} {\bibfnamefont {M.~A.}\ \bibnamefont {Nielsen}}\ and\ \bibinfo {author} {\bibfnamefont {I.~L.}\ \bibnamefont {Chuang}},\ }\href {https://doi.org/10.1017/CBO9780511976667} {\emph {\bibinfo {title} {Quantum Computation and Quantum Information: 10th Anniversary Edition}}}\ (\bibinfo  {publisher} {Cambridge University Press},\ \bibinfo {year} {2010})\BibitemShut {NoStop}%
    \bibitem [{\citenamefont {{Belavkin}}(1975)}]{1975RaEl...20.1177B}%
      \BibitemOpen
      \bibfield  {author} {\bibinfo {author} {\bibfnamefont {V.~P.}\ \bibnamefont {{Belavkin}}},\ }\href@noop {} {\bibfield  {journal} {\bibinfo  {journal} {Radiotekhnika i Elektronika}\ }\textbf {\bibinfo {volume} {20}},\ \bibinfo {pages} {1177} (\bibinfo {year} {1975})}\BibitemShut {NoStop}%
    \bibitem [{\citenamefont {Kholevo}(1979)}]{doi:10.1137/1123048}%
      \BibitemOpen
      \bibfield  {author} {\bibinfo {author} {\bibfnamefont {A.~S.}\ \bibnamefont {Kholevo}},\ }\href {https://doi.org/10.1137/1123048} {\bibfield  {journal} {\bibinfo  {journal} {Theory of Probability \& Its Applications}\ }\textbf {\bibinfo {volume} {23}},\ \bibinfo {pages} {411} (\bibinfo {year} {1979})},\ \Eprint {https://arxiv.org/abs/https://doi.org/10.1137/1123048} {https://doi.org/10.1137/1123048} \BibitemShut {NoStop}%
    \bibitem [{\citenamefont {Hausladen}\ and\ \citenamefont {Wootters}(1994)}]{Hausladen1994AG}%
      \BibitemOpen
      \bibfield  {author} {\bibinfo {author} {\bibfnamefont {P.~A.}\ \bibnamefont {Hausladen}}\ and\ \bibinfo {author} {\bibfnamefont {W.~K.}\ \bibnamefont {Wootters}},\ }\href {https://api.semanticscholar.org/CorpusID:122223374} {\bibfield  {journal} {\bibinfo  {journal} {Journal of Modern Optics}\ }\textbf {\bibinfo {volume} {41}},\ \bibinfo {pages} {2385} (\bibinfo {year} {1994})}\BibitemShut {NoStop}%
    \bibitem [{\citenamefont {Crooks}(2008)}]{PhysRevA.77.034101}%
      \BibitemOpen
      \bibfield  {author} {\bibinfo {author} {\bibfnamefont {G.~E.}\ \bibnamefont {Crooks}},\ }\href {https://doi.org/10.1103/PhysRevA.77.034101} {\bibfield  {journal} {\bibinfo  {journal} {Phys. Rev. A}\ }\textbf {\bibinfo {volume} {77}},\ \bibinfo {pages} {034101} (\bibinfo {year} {2008})}\BibitemShut {NoStop}%
    \bibitem [{\citenamefont {Aw}\ \emph {et~al.}(2021)\citenamefont {Aw}, \citenamefont {Buscemi},\ and\ \citenamefont {Scarani}}]{10.1116/5.0060893}%
      \BibitemOpen
      \bibfield  {author} {\bibinfo {author} {\bibfnamefont {C.~C.}\ \bibnamefont {Aw}}, \bibinfo {author} {\bibfnamefont {F.}~\bibnamefont {Buscemi}},\ and\ \bibinfo {author} {\bibfnamefont {V.}~\bibnamefont {Scarani}},\ }\href {https://doi.org/10.1116/5.0060893} {\bibfield  {journal} {\bibinfo  {journal} {AVS Quantum Science}\ }\textbf {\bibinfo {volume} {3}},\ \bibinfo {pages} {045601} (\bibinfo {year} {2021})}\BibitemShut {NoStop}%
    \bibitem [{\citenamefont {Kwon}\ and\ \citenamefont {Kim}(2019)}]{PhysRevX.9.031029}%
      \BibitemOpen
      \bibfield  {author} {\bibinfo {author} {\bibfnamefont {H.}~\bibnamefont {Kwon}}\ and\ \bibinfo {author} {\bibfnamefont {M.~S.}\ \bibnamefont {Kim}},\ }\href {https://doi.org/10.1103/PhysRevX.9.031029} {\bibfield  {journal} {\bibinfo  {journal} {Phys. Rev. X}\ }\textbf {\bibinfo {volume} {9}},\ \bibinfo {pages} {031029} (\bibinfo {year} {2019})}\BibitemShut {NoStop}%
    \bibitem [{\citenamefont {Buscemi}\ and\ \citenamefont {Scarani}(2021)}]{Buscemi_2021}%
      \BibitemOpen
      \bibfield  {author} {\bibinfo {author} {\bibfnamefont {F.}~\bibnamefont {Buscemi}}\ and\ \bibinfo {author} {\bibfnamefont {V.}~\bibnamefont {Scarani}},\ }\href {https://doi.org/10.1103/PhysRevE.103.052111} {\bibfield  {journal} {\bibinfo  {journal} {Phys. Rev. E}\ }\textbf {\bibinfo {volume} {103}},\ \bibinfo {pages} {052111} (\bibinfo {year} {2021})}\BibitemShut {NoStop}%
    \bibitem [{\citenamefont {Biswas}\ \emph {et~al.}(2023)\citenamefont {Biswas}, \citenamefont {Vaidya},\ and\ \citenamefont {Mandayam}}]{biswas2023noiseadapted}%
      \BibitemOpen
      \bibfield  {author} {\bibinfo {author} {\bibfnamefont {D.}~\bibnamefont {Biswas}}, \bibinfo {author} {\bibfnamefont {G.~M.}\ \bibnamefont {Vaidya}},\ and\ \bibinfo {author} {\bibfnamefont {P.}~\bibnamefont {Mandayam}},\ }\href@noop {} {\  (\bibinfo {year} {2023})},\ \Eprint {https://arxiv.org/abs/2305.11093} {arXiv:2305.11093 [quant-ph]} \BibitemShut {NoStop}%
    \bibitem [{\citenamefont {Shen}\ \emph {et~al.}(2017)\citenamefont {Shen}, \citenamefont {Noh}, \citenamefont {Albert}, \citenamefont {Krastanov}, \citenamefont {Devoret}, \citenamefont {Schoelkopf}, \citenamefont {Girvin},\ and\ \citenamefont {Jiang}}]{PhysRevB.95.134501}%
      \BibitemOpen
      \bibfield  {author} {\bibinfo {author} {\bibfnamefont {C.}~\bibnamefont {Shen}}, \bibinfo {author} {\bibfnamefont {K.}~\bibnamefont {Noh}}, \bibinfo {author} {\bibfnamefont {V.~V.}\ \bibnamefont {Albert}}, \bibinfo {author} {\bibfnamefont {S.}~\bibnamefont {Krastanov}}, \bibinfo {author} {\bibfnamefont {M.~H.}\ \bibnamefont {Devoret}}, \bibinfo {author} {\bibfnamefont {R.~J.}\ \bibnamefont {Schoelkopf}}, \bibinfo {author} {\bibfnamefont {S.~M.}\ \bibnamefont {Girvin}},\ and\ \bibinfo {author} {\bibfnamefont {L.}~\bibnamefont {Jiang}},\ }\href {https://doi.org/10.1103/PhysRevB.95.134501} {\bibfield  {journal} {\bibinfo  {journal} {Phys. Rev. B}\ }\textbf {\bibinfo {volume} {95}},\ \bibinfo {pages} {134501} (\bibinfo {year} {2017})}\BibitemShut {NoStop}%
    \bibitem [{\citenamefont {Iten}\ \emph {et~al.}(2017)\citenamefont {Iten}, \citenamefont {Colbeck},\ and\ \citenamefont {Christandl}}]{PhysRevA.95.052316}%
      \BibitemOpen
      \bibfield  {author} {\bibinfo {author} {\bibfnamefont {R.}~\bibnamefont {Iten}}, \bibinfo {author} {\bibfnamefont {R.}~\bibnamefont {Colbeck}},\ and\ \bibinfo {author} {\bibfnamefont {M.}~\bibnamefont {Christandl}},\ }\href {https://doi.org/10.1103/PhysRevA.95.052316} {\bibfield  {journal} {\bibinfo  {journal} {Phys. Rev. A}\ }\textbf {\bibinfo {volume} {95}},\ \bibinfo {pages} {052316} (\bibinfo {year} {2017})}\BibitemShut {NoStop}%
    \bibitem [{\citenamefont {Rosenblum}\ \emph {et~al.}(2018)\citenamefont {Rosenblum}, \citenamefont {Reinhold}, \citenamefont {Mirrahimi}, \citenamefont {Jiang}, \citenamefont {Frunzio},\ and\ \citenamefont {Schoelkopf}}]{Rosenblum_2018}%
      \BibitemOpen
      \bibfield  {author} {\bibinfo {author} {\bibfnamefont {S.}~\bibnamefont {Rosenblum}}, \bibinfo {author} {\bibfnamefont {P.}~\bibnamefont {Reinhold}}, \bibinfo {author} {\bibfnamefont {M.}~\bibnamefont {Mirrahimi}}, \bibinfo {author} {\bibfnamefont {L.}~\bibnamefont {Jiang}}, \bibinfo {author} {\bibfnamefont {L.}~\bibnamefont {Frunzio}},\ and\ \bibinfo {author} {\bibfnamefont {R.~J.}\ \bibnamefont {Schoelkopf}},\ }\href {https://doi.org/10.1126/science.aat3996} {\bibfield  {journal} {\bibinfo  {journal} {Science}\ }\textbf {\bibinfo {volume} {361}},\ \bibinfo {pages} {266–270} (\bibinfo {year} {2018})}\BibitemShut {NoStop}%
    \bibitem [{\citenamefont {Ma}\ \emph {et~al.}(2020)\citenamefont {Ma}, \citenamefont {Zhang}, \citenamefont {Wong}, \citenamefont {Noh}, \citenamefont {Rosenblum}, \citenamefont {Reinhold}, \citenamefont {Schoelkopf},\ and\ \citenamefont {Jiang}}]{PhysRevLett.125.110503}%
      \BibitemOpen
      \bibfield  {author} {\bibinfo {author} {\bibfnamefont {W.-L.}\ \bibnamefont {Ma}}, \bibinfo {author} {\bibfnamefont {M.}~\bibnamefont {Zhang}}, \bibinfo {author} {\bibfnamefont {Y.}~\bibnamefont {Wong}}, \bibinfo {author} {\bibfnamefont {K.}~\bibnamefont {Noh}}, \bibinfo {author} {\bibfnamefont {S.}~\bibnamefont {Rosenblum}}, \bibinfo {author} {\bibfnamefont {P.}~\bibnamefont {Reinhold}}, \bibinfo {author} {\bibfnamefont {R.~J.}\ \bibnamefont {Schoelkopf}},\ and\ \bibinfo {author} {\bibfnamefont {L.}~\bibnamefont {Jiang}},\ }\href {https://doi.org/10.1103/PhysRevLett.125.110503} {\bibfield  {journal} {\bibinfo  {journal} {Phys. Rev. Lett.}\ }\textbf {\bibinfo {volume} {125}},\ \bibinfo {pages} {110503} (\bibinfo {year} {2020})}\BibitemShut {NoStop}%
    \end{thebibliography}
\end{document}